\documentclass{article}

\usepackage[blocks,affil-it]{authblk}
\usepackage[english]{babel}
\usepackage{amsmath,amssymb,amsthm}
\usepackage{hyperref}
\usepackage{paralist}
\usepackage{bbold}
\usepackage[colorinlistoftodos]{todonotes}
\usepackage{tikz}
\usepackage[figuresright]{rotating}
\usepackage{rotfloat}
\usepackage[ruled,lined,linesnumbered,procnumbered]{algorithm2e}

\usetikzlibrary{decorations.pathreplacing,decorations.pathmorphing,arrows.meta}

\DontPrintSemicolon
\SetKwFor{ForEvery}{forevery}{do}{end}
\SetKw{Var}{Variables}
\SetKw{Int}{Integer}
\SetKwComment{Comment}{}{}

\newtheorem{theorem}{Theorem}[section]
\newtheorem{lemma}[theorem]{Lemma}
\newtheorem{corollary}[theorem]{Corollary}

\newtheorem{example}[theorem]{Example}
\newtheorem{observation}[theorem]{Observation}
\newtheorem{definition}[theorem]{Definition}

\newcommand{\0}{\mathbb{0}}
\newcommand{\1}{\mathbb{1}}
\newcommand{\A}{{\cal A}}
\newcommand{\B}{{\cal B}}

\newcommand{\rA}{\mathrm{A}}
\newcommand{\rH}{\mathrm{H}}
\newcommand{\rC}{\mathrm{C}}
\newcommand{\C}{\mathrm{C}}

\newcommand{\N}{\mathbb{N}}
\newcommand{\R}{\mathrm{R}}

\newcommand{\sem}[1]{[\![#1]\!]}
\newcommand{\sqrun}[1]{\widetilde{#1}}
\newcommand{\T}{\mathrm{T}}

\newcommand{\G}{\mathrm{G}}

\newcommand{\tree}{\mathrm{tree}}

\newcommand{\bb}[1]{\mathbb{#1}}
\newcommand{\cA}{\mathcal{A}}

\newcommand{\bbN}{\mathbb{N}}
\newcommand{\bbB}{\mathbb{B}}
\newcommand{\rR}{\mathrm{R}}
\newcommand{\rT}{\mathrm{T}}

\DeclareMathOperator{\fta}{fta}
\DeclareMathOperator{\hgt}{height}
\DeclareMathOperator{\im}{im}
\DeclareMathOperator{\lcm}{lcm}
\DeclareMathOperator{\past}{past}
\DeclareMathOperator{\pos}{pos}
\DeclareMathOperator{\prefix}{prefix}
\DeclareMathOperator{\rk}{rk}
\DeclareMathOperator{\size}{size}
\DeclareMathOperator{\wt}{wt}

\DeclareMathOperator{\supp}{supp}

\title{Finite-image property of  weighted tree automata over past-finite monotonic strong bimonoids}

\author{Manfred Droste}
\affil{\small Institute of Computer Science,\\ University of Leipzig, Leipzig, Germany}
\author{Zolt{\'a}n F{\"u}l{\"o}p\thanks{Research of this author was supported by grant TUDFO/47138-1/2019-ITM of the Ministry for Innovation and Technology, Hungary.}}
\author{D{\'a}vid K{\'o}sz{\'o}\thanks{Supported by the \'UNKP-20-3 - New National Excellence Program of the Ministry for Innovation and Technology from the source of the National Research, Development and Innovation Fund.}}
\affil{\small Department of Foundations of Computer Science,\\ University of Szeged, Szeged, Hungary}
\author{Heiko Vogler}
\affil{\small Faculty of Computer Science,\\ Technische Universit\"at Dresden, Dresden, Germany}

\begin{document}

\maketitle

\begin{abstract}
  We consider weighted tree automata over strong bimonoids (for short: wta).
  A~wta $\mathcal{A}$ has the finite-image property if its recognized weighted tree language $[\![\mathcal{A}]\!]$ has finite image; moreover, $\mathcal{A}$ has the preimage property if the preimage under $[\![\mathcal{A}]\!]$ of each element of the underlying strong bimonoid  is a recognizable tree language. For each wta $\mathcal{A}$ over a past-finite monotonic strong bimonoid we prove the following results. In terms of $\mathcal{A}$'s structural properties, we characterize whether it has the finite-image property. We characterize those past-finite monotonic strong bimonoids such that
for each wta  $\mathcal{A}$  it is decidable whether  $\mathcal{A}$  has the finite-image property.
In particular, the finite-image property is decidable for wta over past-finite monotonic
semirings. Moreover,  we prove that $\mathcal{A}$  has the preimage property. 
All our results also hold for weighted string automata. 
\end{abstract}

\section{Introduction}

\sloppy Weighted string automata (wsa) were invented \cite{sch61}
for the purpose of describing quantitative properties of recognizable languages, like degree of ambiguity or costs of acceptance. Essentially each wsa is a nondeterministic finite-state automaton in which each transition carries a weight (quantity). In order to calculate with weights, an algebraic structure is needed, called \emph{weight algebra}, and wsa have been investigated over several different weight algebras: semirings \cite{sch61,eil74,berreu88,kuisal86,sak09,drokuivog09}, lattices \cite{wec78,rah09}, strong bimonoids \cite{drostuvog10,cirdroignvog10,drovog12}, valuation monoids \cite{drogoemaemei11,dromei12}, and multi-cost valuation structures \cite{droper16}. The two operations of these weight algebras, usually called addition and multiplication, are used to calculate the weight of a run on a given input word (by means of multiplication) and to sum up the weights of several runs on the given word (by means of addition). In this way, a wsa $\A$ recognizes a weighted language $\sem{\A}$ (or: formal power series), i.e., a mapping from the set of input words to the carrier set of the weight algebra. For the theory of wsa we refer to \cite{sch61,eil74,salsoi78,wec78,berreu82,kuisal86,kui97,sak09,drokuivog09}.

In a similar way, finite-state tree automata have been extended to weighted tree automata (wta) over various weight algebras, e.g.,  complete distributive lattices \cite{inafuk75,esiliu07}, fields \cite{berreu82},
commutative semirings \cite{aleboz87}, strong bimonoids \cite{rad10,drofulkosvog20}, multioperator monoids  \cite{kui99,fulmalvog09,fulstuvog12}, and  tree-valuation monoids \cite{droheuvog15}. In any case, a wta $\A$ recognizes a weighted tree language $\sem{\A}$ (or: formal tree series), i.e., a mapping from the set of input trees to the carrier set of the weight algebra. We note that each wsa is a particular wta (cf. \cite[p.~324]{fulvog09}). For surveys we refer to \cite{esikui03,fulvog09}.

Very important weight algebras for wsa and wta are (a)~the semiring of natural numbers $\N$, (b)~the max-plus-semiring on $\N$, (c)~the min-plus-semiring on $\N$, (d)~the semiring of finite formal languages \cite[Sec.~2]{drokui09}, and (e)~ the semiring of matrices over the positive integers. Apart from (c), these  algebras  are \emph{past-finite} with a suitable order, i.e., each element has only finitely many predecessors in this order. Moreover, the addition and multiplication of each of these algebras are \emph{monotone} with respect to this partial order. In \cite{borfulgazmal05} wta over monotonic semirings were investigated.

Justified by these important examples of weight algebras, we want to advocate in this paper the class of \emph{past-finite monotonic strong bimonoids} as a general model for weight algebras.
These weight algebras share many properties with the semiring of natural numbers: they have  a partial order on its carrier set (which is not necessarily total), they are zero-sum free and zero-divisor free,  their two operations are monotone, and  a strong kind of well-foundedness (called past-finiteness) holds. However, in general, distributivity is not required. 
We will show that the natural numbers $\N$  with addition provide a natural example for past-finite monotonic strong bimonoids which are not semirings.

We will generalize classical results from the theory of wsa and wta over the mentioned specific semirings in this more general setting of weight algebras in a uniform way. We note that classical results for wsa and their proofs crucially employ matrices and therefore need the distributivity of the underlying weight algebras.
Our development uses an analysis of the structure of the wta, but also algebraic means
like congruences and, ultimately, it is combined with a reduction to the classical results for the semiring of natural numbers.

There are two  natural questions associated with a wta $\A$:
\begin{compactitem}
\item Does $\A$ have the finite-image property? A wta $\A$ has the \emph{finite-image property} if the weighted tree language $\sem{\A}$ has finite image.
  \item Does $\A$ have the preimage property? A wta $\A$ has the \emph{preimage property} if the preimage under  $\sem{\A}$ of each element of the weight algebra is a recognizable tree language (cf., e.g., \cite{eng75-15,gecste84} for the theory of recognizable tree languages). 
  \end{compactitem}

  In the literature there are  some answers to these questions. 
  Each wsa over a finite semiring and over the semiring of natural numbers
has the preimage property, and each wsa over a commutative ring
which has the finite-image property also has the preimage property \cite[Ch.~III]{berreu88}.
This has also been shown for wta \cite{drovog06,lou99}.  Furthermore, for each wsa over
any subsemiring of the rational numbers, the finite-image property is decidable \cite{mansim77}.
This latter property is related to the classical Burnside property for semigroups \cite{resreu85}.
Moreover, each wta over a locally finite semiring  \cite{drovog06} and each wsa over a bi-locally finite strong bimonoid \cite{drostuvog10}  has the finite-image property and the preimage property. Thus, in particular, each wsa over a bounded (not necessarily distributive) lattice has the two properties.

Weighted tree languages which are recognized by wta that have both, the finite-image property and the preimage property, are called \emph{recognizable step mappings} \cite{drogas05,drovog06}.  The class of such mappings is characterized by crisp-deterministic wta \cite{fulkosvog19} (cf. \cite{drostuvog10} for the string case). Intuitively, such a wta  can be considered as a usual (unweighted) deterministic finite-state tree automaton in which each final state carries a weight.

In this paper, we investigate the two mentioned questions for wta over past-finite monotonic strong bimonoids.
It is an extended version of \cite{drofulkosvog20} and our main results are the following:
\begin{itemize}
\item For each wta $\A$ over some arbitrary strong bimonoid, we give a sufficient criterion such that $\A$ has the finite-image property and the preimage property (cf. Theorem \ref{thm:sufficient-conds-ensure-rec-step-map}).

\item Each wta $\A$ over some past-finite monotonic strong bimonoid has the preimage property (cf. Theorem~\ref{thm:inverse_b_recognizable}).

\item For each wta $\A$  over some past-finite monotonic strong bimonoid, we characterize when $\A$ has the finite-image property, in terms of structural properties of $\A$ (cf. Theorem~\ref{cor:characterization-of-crisp-determinizability}).

\item We characterize the subclass $\mathcal C$ of those past-finite monotonic strong bimonoids
      for which the following holds: for each wta over some weight algebra from $\mathcal C$, it is decidable whether it has the finite-image property (cf. Theorem~\ref{thm:characterization}). In particular, $\mathcal C$ contains all past-finite monotonic semirings (cf. Theorem~\ref{thm:past-fin-mon-semiring-fin-im-dec}).

      \item Given a wta $\A$  over a past-finite monotonic strong bimonoid and some  $k \in \N_+$,
it is decidable whether the cardinality of the image of $\sem{\A}$  is bounded by  $k$ (cf. Theorem \ref{thm:dec-imA-leq-k}). 
  \end{itemize}
  All the above results except Item 2 are new, i.e., do not appear in \cite{drofulkosvog20}.
For our decidability results, we assume that the respective weight algebras are given
in a computable way.

Since wsa \cite{sch61,eil74} over semirings are a special case of wta over semirings (cf. \cite[p.~324]{fulvog09}), and this relationship also holds for wsa over strong bimonoids, all our results also hold for wsa. We will explain this in more detail in Section \ref{sect:string-case}.

\section{Preliminaries} \label{sect:preliminaries}

\subsection{General notions and notations}

We denote by $\mathbb{N}$ the set of natural numbers $\{0,1,2,\ldots\}$ and by $\N_+$ the set $\N \setminus \{0\}$. 
For every $m,n \in \N$, we denote the set $\{i \in \N \mid m \leq i \leq n\}$ by $[m,n]$. We abbreviate $[1,n]$ by $[n]$. Hence, $[0]=\emptyset$.

Let $A$ be a set. Then $|A|$ denotes the cardinality of $A$, ${\cal P}_{\mathrm{f}}(A)$ denotes the set of finite subsets of $A$,  $A^*$ denotes the {\em set of all strings over $A$}, and $\varepsilon$ denotes the empty string. For every $v,w \in A^*$, $vw$ denotes the {\em concatenation of $v$ and $w$}, $|v|$ denotes the length of $v$, and $\prefix(v)$ denotes the set $\{w \in A^* \mid (\exists u \in A^*): v=wu\}$.

Let $B$ be a set and $R$ a binary relation on $B$. As usual, for every $a,b\in B$, we write $a R b$ instead of $(a,b)\in R$. We call $R$ an {\em equivalence relation} if it is reflexive, symmetric, and transitive.
If $R$ is an equivalence relation, then for each $b\in B$ we denote by $[b]_R$ the equivalence class $\{a\in B \mid aRb\}$ and by $B/_{R}$ the set $\{[b]_R\mid b\in B\}$. We say that  $R$ is a {\em partial ordering} if it is reflexive, antisymmetric, and transitive.
For each $b \in B$, let  $\past(b) = \{a \in B \mid a R b\}$.
We call $(B,R)$ \emph{past-finite} if $\past(b)$ is finite for each $b \in B$.

Let $f:A\to B$ and $g:B\to C$ be mappings, where $C$ is a further set.
The {\em image of $f$} is the set $\im(f)=\{f(a) \mid a \in A\}$, and for each $b \in B$,  the \emph{preimage of $b$ under $f$} is the set $f^{-1}(b)=\{a \in A \mid f(a)=b\}$. Moreover, the {\em composition of $f$ and $g$} is the mapping $g\circ f: A \to C$ defined by $(g\circ f)(a)=g(f(a))$ for each $a\in A$.

\subsection{Trees and contexts}

We suppose that the reader is familiar with the fundamental concepts and results of the theory of finite-state tree automata and recognizable tree languages \cite{eng75-15,gecste84,tata08}. Here we only recall some basic definitions. 

A {\em ranked alphabet} is a tuple $(\Sigma,\rk)$ which consists of an alphabet $\Sigma$ and mapping $\rk: \Sigma \to \N$, called {\em rank mapping}, such that $\rk^{-1}(0)\ne \emptyset$. For each $k \in \N$, we define $\Sigma^{(k)}=\{\sigma \in \Sigma \mid \rk(\sigma) = k\}$. Sometimes we write $\sigma^{(k)}$ to indicate
that $\sigma \in \Sigma^{(k)}$.
%
%
As usual, we abbreviate $(\Sigma,\rk)$ by $\Sigma$ if $\rk$ is irrelevant or it is clear from the context.
If $\Sigma=\Sigma^{(1)}\cup \Sigma^{(0)}$ such that $|\Sigma^{(1)}|\ge 1$ and $|\Sigma^{(0)}|=1$, then we call $\Sigma$ a \emph{string ranked alphabet}.

Let $\Sigma$ be a ranked alphabet and $H$ be a set disjoint from $\Sigma$.
The {\em set of $\Sigma$-trees over $H$}, denoted by $\T_\Sigma(H)$, is the smallest set $T$ such that (i) $\Sigma^{(0)} \cup H \subseteq T$ and (ii) if $k \in \N_+$, $\sigma \in \Sigma^{(k)}$, and $\xi_1,\ldots,\xi_k \in T$, then $\sigma(\xi_1,\ldots,\xi_k) \in T$.
We write  $\T_\Sigma$ for $\T_\Sigma(\emptyset)$.
For every $\gamma \in \Sigma^{(1)}$ and $\alpha \in \Sigma^{(0)}$, we abbreviate the tree $\gamma(\ldots \gamma(\alpha) \ldots)$ with $n$ occurrences of $\gamma$ by $\gamma^n(\alpha)$ and write $\gamma$ for $\gamma^1$.
Any subset $L$ of $\T_\Sigma$ is called \emph{$\Sigma$-tree language}.

We define the {\em set of positions} of trees as a mapping $\pos: \T_\Sigma(H) \to {\cal P}_{\mathrm{f}}(\N_+^*)$ such that (i) for each $\xi \in (\Sigma^{(0)} \cup H)$ let $\pos(\xi)=\{\varepsilon\}$ and (ii) for every $\xi=\sigma(\xi_1,\ldots,\xi_k)$ with $k \in \N_+$, $\sigma\in \Sigma^{(k)}$, and $\xi_1,\ldots,\xi_k\in \T_\Sigma(H)$, let $\pos(\xi)=\{\varepsilon\}\cup\{iv \mid i \in [k], v \in \pos(\xi_i)\}$. The {\em height} and the {\em size} of a tree $\xi\in \T_\Sigma$ are $\hgt(\xi)=\max\{|v| \mid v\in \pos(\xi)\}$ and $\size(\xi)=|\pos(\xi)|$, respectively.

Let $\xi, \zeta \in \T_\Sigma(H)$ and $v \in \pos(\xi)$.
Then the {\em label of $\xi$ at $v$}, denoted by $\xi(v)$, the {\em subtree of $\xi$ at $v$}, denoted by $\xi|_v$, and the {\em replacement of the subtree of $\xi$ at $v$ by $\zeta$}, denoted by $\xi[\zeta]_v$, are defined as follows: 
\begin{compactitem}
\item[(i)] if $\xi \in (\Sigma^{(0)} \cup H)$, then we let $\xi(\varepsilon)=\xi$, $\xi|_\varepsilon = \xi$, and $\xi[\zeta]_\varepsilon=\zeta$ and 
\item[(ii)] for every $\xi=\sigma(\xi_1,\ldots,\xi_k)$ with $k \in \N_+$, $\sigma \in \Sigma^{(k)}$, and $\xi_1,\ldots,\xi_k \in \T_\Sigma(H)$,
  we define $\xi(\varepsilon)=\sigma$ and $\xi|_\varepsilon = \xi$, and $\xi[\zeta]_\varepsilon = \zeta$, and for every $i \in [k]$ and $v' \in \pos(\xi_i)$, we define 
\begin{compactitem}
\item $\xi(iv')=\xi_i(v')$,
\item $\xi|_{iv'}=\xi_i|_{v'}$, and
\item $\xi[\zeta]_{iv}=\sigma(\xi_1,\ldots,\xi_{i-1},\xi_i[\zeta]_v, \xi_{i+1}, \ldots, \xi_k)$.
\end{compactitem}
\end{compactitem}

Let $\square$ be a symbol such that $\square \not\in \Sigma$. For each $\zeta\in \T_\Sigma(\{\square\})$,  we define $\pos_\square(\zeta)=\{v \in \pos(\zeta) \mid \zeta(v) = \square\}$, and for each $v \in \pos(\zeta)$ we abbreviate by $\zeta|^v$ the tree $\zeta[\square]_v$. We denote by $\C_\Sigma$ the set $\{\zeta \in \T_\Sigma(\{\square\}) \mid |\pos_\square(\zeta)|=1\}$, and we call its elements {\em contexts over $\Sigma$} (for short: {\em $\Sigma$-contexts} or  {\em contexts}). Thus a context is a tree over the ranked alphabet $\Sigma$ and the set $H=\{ \square \}$ in which $\square$ occurs precisely once, as a leaf.

Let $c\in \C_\Sigma$ with $\{v\}=\pos_\square(c)$ and $\zeta \in (\T_\Sigma\cup \C_\Sigma)$. Then we abbreviate $c[\zeta]_v$ by $c[\zeta]$. Hence $c[\zeta]$  is obtained from the context  $c$
by replacing the leaf  $\square$  by $\zeta$. Obviously, if $\zeta \in  \C_\Sigma$, then also $c[\zeta]\in  \C_\Sigma$. Moreover, for each $n \in \N$, we define the {\em $n$th power of $c$}, denoted by $c^n$, by induction as follows: $c^0=\square$ and  $c^{n+1}=c[c^n]$.

\begin{quote}
\em In the rest of this paper, $\Sigma$ will denote an arbitrary ranked alphabet if not specified otherwise. Moreover, if we write `$\xi=\sigma(\xi_1,\ldots,\xi_k)$', then we mean that there exist $k \in \N_+$, $\sigma \in \Sigma^{(k)}$, and $\xi_1,\ldots,\xi_k \in \T_\Sigma(H)$ such that $\xi = \sigma(\xi_1,\ldots,\xi_k)$.
\end{quote}

\subsection{Strong bimonoids}

A {\em strong bimonoid} \cite{drostuvog10,cirdroignvog10,drovog12} is an algebra $(B,\oplus,\otimes,\0,\1)$ such that $(B,\oplus,\0)$ is a commutative monoid, $(B,\otimes,\1)$ is a monoid, $\0 \neq \1$, and $\0 \otimes b = b \otimes \0=\0$ for each $b \in B$. 

\ 

We say that $B$ is 
\begin{compactitem}
    \item \emph{commutative} if $\otimes$ is commutative,
    \item {\em left distributive} (respectively, {\em right distributive}) if $\otimes$ is distributive over $\oplus$ from the left (respectively, the right), and
    \item a {\em semiring} if it is left and right distributive.
\end{compactitem} 
Moreover, we call $B$
\begin{compactitem}
\item {\em one-product free} if $a \otimes b = \1$ implies $a = \1 = b$, 
\item {\em zero-divisor free} if $a \otimes b = \0$ implies $a = \0$ or $b = \0$,
\item {\em zero-sum free} if $a \oplus b = \0$ implies $a = \0$ and $b = \0$, and
\item {\em (additively) idempotent} if $a \oplus a = a$
\end{compactitem}
for every $a,b \in B$.

In \cite[Def.~12]{borfulgazmal05} the concept of {\em monotonic semiring} is defined. In the spirit of this definition, we define monotonic strong bimonoids as follows.

\begin{definition} \rm
\sloppy Let $(B,\oplus,\otimes,\0,\1)$ be a strong bimonoid and $\preceq$ a partial order on $B$. We write $a \prec b$ to denote that $a \preceq b$ and $a \ne b$.  We say that  $(B,\oplus,\otimes,\0,\1,\preceq)$ is {\em monotonic} if the following conditions hold:
\begin{compactitem}
\item[(i)]  for every $a, b \in B: a \preceq a \oplus b$, and
\item[(ii)] for every  $a,b,c \in B \setminus \{\0\}$  with  $b \ne \1$  we have:  $a \otimes c \prec a \otimes b \otimes c$.
\end{compactitem}
We call $(B,\oplus,\otimes,\0,\1,\preceq)$ \emph{past-finite} if $(B,\preceq)$ is past-finite.
\end{definition}

If  $B$  is monotonic, then, as is easy to check, $\0 \prec \1 \prec b$  for each  $b \in B \setminus \{\0,\1\}$;
hence, clearly  $B$  is one-product free, zero-divisor free and zero-sum free.
Moreover, if  $B$  has at least 3 elements, then  $B$  is infinite. The only monotonic strong bimonoid
with 2 elements is the Boolean semiring $(\mathbb{B},\sup,\inf,0,1)$ with its natural order, where $\mathbb{B}=\{0,1\}$. Cf. \cite[p.~122]{borfulgazmal05} for further results on
monotonic semirings.

\begin{example} \rm
We give six examples of past-finite monotonic  semirings (cf. \cite[p.~122]{borfulgazmal05}):
\begin{compactenum}
\item[(i)] \sloppy  the semiring of natural numbers \mbox{$(\N,+,\cdot,0,1,\leq)$};
\item[(ii)] the arctic semiring $\mathrm{ASR}_\N=(\N_{-\infty},\max,+,-\infty,0,\leq)$, where $\N_{-\infty} = \N \cup \{-\infty\}$; 
\item[(iii)] the semiring $\mathrm{Lcm}=(\N,\lcm,\cdot,0,1,\leq)$, where $\lcm(0,n)=n=\lcm(n,0)$ for each $n \in \N$ and otherwise $\lcm$ is the usual least common multiple;
\item[(iv)] the semiring $\mathrm{FSet}(\N) = ({\cal P}_{\mathrm{f}}(\N),\cup,+, \emptyset,\{0\},\preceq)$ where the addition on $\N$ is extended to sets as usual, and $\preceq$ is defined by $N_1\preceq N_2$ if there is an injective mapping $f:N_1\to N_2$ such that $n\le f(n)$ for each $n\in N_1$; 
\item[(v)] for each $n \in \N_+$, the semiring $\mathrm{Mat}_n(\N_+) = (\N_+^{n \times n} \cup \{\underline{0}, \underline{1}\},+,\cdot,\underline{0}, \underline{1},\leq)$ of square matrices over $\N_+$ with the common matrix addition and multiplication, where $\underline{0}$ is the $n \times n$ zero matrix and $\underline{1}$ is the $n \times n$ unit matrix; the partial order $\leq$  is defined  by $M\le M'$  if $M_{ij}\le M'_{ij}$ for each $(i,j)\in[n]\times[n]$; and 
\item[(vi)] the semiring $\mathrm{FLang}_\Sigma = ({\cal P}_{\mathrm{f}}(\Sigma^*), \cup, \cdot, \emptyset, \{\varepsilon\},\preceq)$ over the alphabet $\Sigma$ with the operations of union and concatenation, and $\preceq$ is defined by $L_1\preceq L_2$ if there is an injective mapping $f:L_1\to L_2$ such that $w$ is a subword of $f(w)$ for each $w\in L_1$. 
\end{compactenum}
The semirings in (i)-(iv) are commutative, and the semirings in (ii)-(iv) and (vi) are idempotent. \qed
\end{example}

Next we wish to give a natural example of a past-finite monotonic strong bimonoid which is not a semiring.

\begin{example} \label{ex:plus-plus-sb} \rm
  Take $(\N_\0,\oplus,+,\0,0)$, the natural numbers with plus and plus, with a new zero $\0$ added. That is, we have $\N_\0 = \N \cup \{ \0 \}$, the
  bimonoid addition on  $\N$  is the usual one, denoted by $\oplus$, the bimonoid multiplication on  $\N$  is also the usual addition, denoted by $+$, in order to indicate that here the usual addition serves as bimonoid multiplication. Moreover,
  $\0 \oplus x = x$  and  $\0 + x = \0$  for each $x \in \N_\0$. Let $\leq$  be the usual order on  $\N$  together with
$\0 < x$  for each  $x \in \N$. Then  $(\N_\0,\oplus,+,\0,0,\leq)$  is a past-finite strong bimonoid which is not a semiring. We might call this structure the \em{plus-plus-strong bimonoid of natural numbers}.
\end{example}

Now we give an example of an additively locally finite and  past-finite monotonic strong bimonoid which is not a semiring.

\begin{example} \label{ex:add-loc-fin-past-fin} \rm
The strong bimonoid $(\N,+',\cdot,0,1,\leq)$ with the operation $+'$ defined, for each $a, b \in \N$, by
\[
a +' b =
\begin{cases}
\min\{a + b, 100\} &\text{if $a, b \leq 100$}\\
\max\{a, b\} &\text{otherwise}
\end{cases}
\]
is additively locally finite and past-finite monotonic. Moreover, it is neither left distributive nor right distributive. \qed
\end{example}

Next we show a general method for generating past-finite monotonic strong bimonoids.

\begin{example}\label{ex:general-method} \rm
Let $(B,\preceq)$ be a past-finite partially ordered set. Let $(B,+)$ be a commutative semigroup such that,  for every $a,b \in B$, we have $a \preceq a + b$. Moreover, let $(B,\times)$ be a semigroup such that,  for every $a, b, c \in B$, we have $a \prec a \times b$, $c \prec b \times c$, and  $a \times c \prec a \times b \times c$.

According to \cite[Ex.~2.1(4)]{drovog12}, we construct the \emph{strong bimonoid induced by  $(B,+)$ and $(B,\times)$} to be the strong bimonoid $(B',\oplus,\otimes,\0,\1)$ defined as follows:  
\begin{compactitem}
    \item $B' = B \cup \{\0,\1 \}$ where $\0,\1 \not\in B$;
    \item
 we define the operation $\oplus: B' \times B' \to B'$ such that $\oplus|_{B \times B} = +$ and
       for each $b \in B'$ we let $\0 \oplus b = b = b \oplus \0$ and, if $b \neq \0$, then $\1 \oplus b = b = b \oplus \1$;
    \item
       we define the operation $\otimes: B' \times B' \to B'$ such that $\otimes|_{B \times B} = \times$ and
         for each $b \in B'$ we let $\0 \otimes b = \0 = b \otimes \0$, and $\1 \otimes b = b = b \otimes \1$.
   \end{compactitem}
We define the partial ordering $\preceq'$ on $B'$ such that
$\0 \prec' \1 \prec' b$  for each $b \in B$ and $\preceq' \cap (B \times B) =\,\preceq$.
Then $(B',\oplus,\otimes,\0,\1,\preceq')$ is past-finite monotonic.

To verify this, we make the following observations.
Clearly,  $(B',\preceq')$ is past-finite. By case analysis, it is easy to show that $\oplus$ and $\otimes$ satisfy  properties (i) and (ii) of the definition of monotonic strong bimonoid, respectively. In particular, property (ii) can be seen as follows. Let $a,b,c \in B' \setminus \{\0\}$ and $b \ne \1$ (i.e., $b \in B$):
We claim that  $a \otimes c \prec' a \otimes b \otimes c$.

If  $a = c = \1$, then $a \otimes c = \1 \prec' b = a \otimes b \otimes c$.

If  $a \ne \1$  and  $c = \1$, then we have  $a \in B$  and  $a \prec a \times b$  by the assumption on $(B,\times)$, and hence $a \prec' a \times b$ by the definition of $\prec'$. Then $a \otimes c = a \prec' a \times b = a \otimes b = a \otimes b \otimes c$.

If  $a = \1$ and  $c \ne \1$, then $c \in B$  and  $c \prec b \times c$  by the assumption on $(B,\times)$. Using the same arguments as in the previous case, we have $a \otimes c = c \prec' b \times c = b \otimes c = a \otimes b \otimes c$.

Finally, if $a \ne \1$ and $b \ne \1$, then we have  $a,b,c \in B$ and therefore
$a \times c \prec a \times b \times c$  by the assumption on $(B,\times)$. Since $\otimes$ on $B$ and $\prec'$ on $B$ are equal to $\times$ and $\prec$, respectively, we obtain  $a \otimes c \prec' a \otimes b \otimes c$. 

Clearly, $(B',\oplus,\otimes,\0,\1)$ is additively locally finite if, for each finite $A \subseteq B$, the subsemigroup of $(B,+)$ generated by $A$ is finite.
\qed
\end{example}

As an application of the general method, we let $(B,\preceq) = (\N_+,\leq)$,
$(B,+) = (\N_+,+)$  and  $(B,\times) = (\N_+,+)$, both with the usual addition of natural numbers, and $\1 = 0$. Then we obtain the plus-plus-strong bimonoid of
natural numbers given in Example~\ref{ex:plus-plus-sb}. As another application, we can choose $(B,\preceq)=(\mathbb{N}_+,\leq)$, $(B,+)=(\mathbb{N}_+,+)$ with the usual addition on natural numbers, and $(B,\times) = (\mathbb{N}_+,\times)$ where $a \times b = a + b + 2ab$ for every $a,b \in \N_+$. Then $B'$ is a past-finite monotonic strong bimonoid.   Moreover, $B'$ is neither left distributive nor right distributive. As a third one, we can consider the semigroup $(B,+)=(\mathbb{N}_+,\max)$ and the rest as above. Then  $B'$ is an idempotent and past-finite monotonic semiring.

For later use, we introduce the following notations and notions.

We extend $\oplus$ to every finite set $I$ and family $(b_i \mid i \in I)$ of elements of $B$ as usual, and denote the extended operation by $\bigoplus$.
We always abbreviate $\bigoplus (b_i \mid i \in I)$ by $\bigoplus_{i \in I}b_i$. Moreover, if $I=[k]$ for some $k \in \N$, then we write $\bigoplus_{i=1}^k b_i$. For each $b\in B$,  we abbreviate $\bigoplus_{i=1}^k b$ by $kb$.

We also extend the operation $\otimes$ to every $k \in \N$ and family $(b_i \mid i \in [k])$ of elements of $B$ as usual, and denote the extended operation by $\bigotimes$.
We always abbreviate $\bigotimes (b_i \mid i \in [k])$ by $\bigotimes_{i=1}^k b_i$. For each $b\in B$,  we abbreviate $\bigotimes_{i=1}^k b$ by $b^k$.
For each $B_1,B_2 \subseteq B$, we denote by $B_1 \otimes B_2$ the set $\{a \otimes b \mid a \in B_1, b \in B_2\}$.

For each $A\subseteq B$, we denote by $\langle A\rangle_\oplus$ the \emph{submonoid of $(B,\oplus,\0)$ generated by $A$} and we denote by $\langle A\rangle_\otimes$  the \emph{submonoid of $(B,\otimes,\1)$ generated by $A$}.
We say that $B$ is \emph{additively locally finite} (respectively, \emph{multiplicatively locally finite}) if, for each finite $A \subseteq B$, the submonoid $\langle A\rangle_\oplus$ (respectively, $\langle A\rangle_\otimes$) is finite. Observe that, if $B$ is idempotent, then it is also additively locally finite. We say that $B$ is \emph{bi-locally finite} if it is both additively locally finite and 
multiplicatively locally finite.
For each $b \in B$, we abbreviate $\langle \{b \}\rangle_\oplus$ by $\langle b \rangle_\oplus$ and
if $\langle b \rangle_\oplus$ is finite, then we say that $b$ has {\em finite additive order}.

\begin{quote}
\em In the rest of the paper, $(B,\oplus,\otimes,\0,\1)$ denotes an arbitrary strong bimonoid if not specified otherwise. 
\end{quote}

\subsection{Weighted tree languages}\label{sect:wtl}

A {\em $(\Sigma,B)$-weighted tree language} (for short: weighted tree language) is a mapping $r:\T_\Sigma \to B$. For every $(\Sigma,B)$-weighted tree language $r$, we denote by $\supp(r)$ the set $\{\xi \in \T_\Sigma \mid r(\xi) \neq \0\}$.

Let $L \subseteq \T_\Sigma$ be a $\Sigma$-tree language. The {\em characteristic mapping of $L$ with respect to $B$} is the mapping $\1_{(B,L)}: \T_\Sigma \to B$ defined, for each $\xi \in \T_\Sigma$, by $\1_{(B,L)}(\xi) = \1$ if $\xi \in L$ and $\0$ otherwise.

Let $r$ and $r'$ be $(\Sigma,B)$-weighted tree languages and $b \in B$. We define the $(\Sigma,B)$-weighted tree languages $r \oplus r'$ and $b \otimes r$, for each $\xi \in \T_\Sigma$, by $(r \oplus r')(\xi) = r(\xi) \oplus r'(\xi)$ and $(b \otimes r)(\xi) = b \otimes r(\xi)$, respectively.

We say that $r$ is a {\em $(\Sigma,B)$-recognizable one-step mapping} (or just: recognizable one-step mapping) if there exist  a recognizable $\Sigma$-tree language $L \subseteq \T_\Sigma$ and a $b\in B$ such that $r = b \otimes \1_{(B,L)}$. The tree language $L$ is called {\em step language}. Moreover, $r$ is a {\em $(\Sigma,B)$-recognizable step mapping} (or just: recognizable step mapping) if there exist  $n \in \N_+$ and $(\Sigma,B)$-recognizable one-step mappings $r_1,\ldots, r_n$ such that $r=\bigoplus_{i = 1}^n r_i$ (where we extend the sum of two weighted tree languages in a straightforward way to the sum of finitely many weighted tree languages). Obviously, if $r$ is a recognizable step-mapping, then $\im(r)$ is finite. We note that, in \cite{bor04}, recognizable one-step mappings were called weighted tree languages which are constant on their supports.

For each $b\in B$, we define the weighted tree language $\widetilde{b}$ by $\widetilde{b}(\xi)=b$ for each $\xi \in \T_\Sigma$. We note that $\widetilde{b}$ is the recognizable one-step mapping $b \otimes \1_{(B,\T_\Sigma)}$.

\section{Weighted tree automata with run semantics}\label{sect:wta-run-sem}

We recall the concept of weighted tree automata over strong bimonoids from \cite{rad10} (also cf., e.g., \cite{fulvog09,fulkosvog19}).
A \emph{weighted tree automaton (over $\Sigma$ and  $B$)}  (for short: $(\Sigma,B)$-wta or wta) 
is a tuple $\A=(Q,\delta,F)$, where $Q$ is a finite nonempty set ({\em states}), $\delta=(\delta_k \mid k \in \N)$ is a family of mappings $\delta_k: Q^k \times \Sigma^{(k)} \times Q \to B$ ({\em transition mappings}), and $F: Q \to B$ ({\em root weight mapping}). 

\begin{quote}
\em From now on, for every $k\in \N$, $(q_1,\ldots,q_k) \in Q^k$, $\sigma \in \Sigma^{(k)}$, and $q \in Q$, we abbreviate expressions of the form $\delta_k((q_1,\ldots,q_k),\sigma,q)$ by $\delta_k(q_1 \cdots q_k, \sigma,q)$. Moreover, we write $F_q$ instead of $F(q)$ for each $q\in Q$. 
\end{quote}

We say that $\A$ is \emph{deterministic} (and {\em crisp-deterministic}) if, for every $k \in \N$, $w \in Q^k$, and $\sigma \in \Sigma^{(k)}$ there exists at most one $q \in Q$ such that $\delta_k(w,\sigma,q)\ne \0$ (respectively, there exists a $q \in Q$ such that $\delta_k(w,\sigma,q)=\1$, and $\delta_k(w,\sigma,q')=\0$ for each $q' \in Q \setminus \{q\}$). Clearly, crisp-determinism implies determinism.

We define the run semantics for a  $(\Sigma,B)$-wta as follows. 
Let $\A=(Q,\delta,F)$ be a $(\Sigma,B)$-wta, $\zeta \in \T_\Sigma(\{\square\})$, and $\rho: \pos(\zeta) \to Q$. We call $\rho$ a {\em run of $\A$ on $\zeta$} if, for every $v \in \pos(\zeta)$ with $\zeta(v) \in \Sigma$, we have $\delta_{k}(\rho(v1) \cdots \rho(vk), \sigma, \rho(v)) \neq \0$ where $\sigma = \zeta(v)$ and $k= \rk(\sigma)$.
  If $\rho(\varepsilon)=q$ for some $q \in Q$, then we say that $\rho$ is a {\em $q$-run on $\zeta$}. We denote by $\R_\A(q,\zeta)$ the set of all $q$-runs on $\zeta$ and we let $\R_\A(\zeta)=\bigcup_{q\in Q}\R_\A(q,\zeta)$.  If $\A$ is deterministic, then $|\R_\A(\zeta)| \le 1$. Moreover, we let $\rR_\cA^{F \neq \bb0}(\zeta)$ denote the set of all $\rho \in \rR_\cA(\zeta)$ such that $F_{\rho(\varepsilon)} \neq \bb0$.  In particular, for $c \in \C_\Sigma$ with $\pos_\Box(c)=\{v\}$, we call each $\rho \in \R_\A(q,c)$ a \emph{$(q,\rho(v))$-run on 
$c$} and we denote the set of all $(q,p)$-runs on $c$ by $\R_\A(q,c,p)$.  We note that $\R_\A(q,c)=\bigcup_{p \in Q}\R_\A(q,c,p)$. Each element of $\R_\A(q,c,q)$ is called \emph{loop}.

Let $\zeta \in \T_\Sigma(\{\square\})$, $\rho \in \R_\A(\zeta)$, and $v \in \pos(\zeta)$. We define the mapping $\rho|_v: \pos(\zeta|_v) \to Q$ such that, for each $v' \in \pos(\zeta|_v)$, we have $\rho|_v(v')=\rho(vv')$. Clearly, $\rho|_v \in \R_\A(\zeta|_v)$, and hence we call it the {\em run induced by $\rho$ at $v$}.

We say that $\cA$ is {\em finitely ambiguous} if there exists $K \in \bbN$ such that,  for each $\xi \in \rT_\Sigma$, we have  $|\rR_\cA^{F \neq \bb0}(\xi)| \leq K$. Moreover, we call $\cA$ {\em unambiguous} if it is finitely ambiguous and $K=1$. We note that each deterministic wta is unambiguous, and there exist  easy examples of unambiguous wta  for which there does not exist an equivalent deterministic wta \cite{klilommaipri04}.

Now we define the {\em weight of a run $\rho \in \R_\A(\zeta)$} to be the element $\wt_\A(\zeta, \rho)$ of $B$ by induction as follows: (i) if $\zeta = \square$, then $\wt_\A(\zeta,\rho)=\1$ and (ii) if $\zeta=\sigma(\zeta_1,\ldots,\zeta_k)$ then $\wt_\A(\zeta,\rho)$ is defined by
\begin{equation} \label{eq:weight-of-a-run}
\wt_\A(\zeta,\rho) = \Big(\bigotimes_{i=1}^k \wt_\A(\zeta_i,\rho|_i)\Big) \otimes \delta_k\big(\rho(1)\cdots \rho(k),\sigma,\rho(\varepsilon)\big)\enspace.
\end{equation}
If there is no confusion, then we drop the index $\A$ from $\wt_\A$ and write just $\wt(\zeta, \rho)$ for the weight of $\rho$.

The {\em run semantics of $\A$} is the $(\Sigma,B)$-weighted tree language $\sem{\A}:\T_\Sigma \to B$ defined, for each $\xi \in \T_\Sigma$, by
\[
\sem{\A}(\xi)
= \bigoplus_{\rho \in \R_\A(\xi)} \wt(\xi, \rho) \otimes F_{\rho(\varepsilon)} = \bigoplus_{\rho\in\rR_\cA^{F \neq \bb0}(\xi)}\wt(\xi)\otimes F_\rho(\varepsilon)\enspace,
\]
where the second equality holds because $\wt(\xi,\rho)\otimes F_{\rho(\varepsilon)}=\0$ for each $\rho \in  \rR_\cA(\xi)\setminus \rR_\cA^{F \neq \bb0}(\xi)$. We will use the above equality without any reference.
Let $\A$ and $\B$ be $(\Sigma,B)$-wta. We say that $\A$ and $\B$ are \emph{equivalent} if $\sem{\A} = \sem{\B}$. A weighted tree language $r:\T_\Sigma \to B$ is {\em run-recognizable} if there exists a $(\Sigma,B)$-wta $\A$ such that $r=\sem{\A}$. The class of all run-recognizable $(\Sigma,B)$-weighted tree languages is denoted by $\mathrm{Rec}(\Sigma,B)$.

\begin{example} \rm
For the weight structure of Example~\ref{ex:plus-plus-sb}, the plus-plus-strong bimonoid of natural numbers, the weighted tree automaton along a run would sum up all weights (costs) of the transitions occurring
in the run, but to determine the weight of a tree it would also execute all possible runs and sum up their weights (costs).
This might be considered as the total sum of the weights of all transitions of all non-deterministic realizations (runs).
\end{example}

The following fact is well known and we will use it in the paper without any further reference (cf. e.g. \cite[Sect.~3.4]{fulvog09}).
 A $\Sigma$-tree language $L \subseteq \T_\Sigma$ is \emph{recognizable}, i.e., recognizable by a finite-state $\Sigma$-tree automaton if and only if there exists a $(\Sigma,\mathbb{B})$-wta $\A$ such that $L= \supp(\sem{\A})$ (recall that $\mathbb{B}$ is the Boolean semiring).
Moreover, for each $(\Sigma,\mathbb{B})$-wta $\A$, we can construct a finite-state $\Sigma$-tree automaton which recognizes $\supp(\sem{\A})$. Vice versa, for each finite-state $\Sigma$-tree automaton which recognizes $L$, we can construct a $(\Sigma,\mathbb{B})$-wta $\A$ such that $L= \supp(\sem{\A})$. \emph{Therefore, in order to avoid using several automata models, we identify finite state $\Sigma$-tree automata  with $(\Sigma,\bbB)$-wta.}

 We note that also another semantics, called {\em initial algebra semantics}, can be defined for $\A$ \cite{rad10,fulvog09,fulkosvog19}. In general, the two kinds of semantics are different \cite{drostuvog10}, however, if $B$ is a semiring or $\A$ is deterministic, then they coincide \cite[Lm.~4.1.13]{bor04b}, \cite[Thm.~4.1]{rad10}, and \cite[Thm.~3.10]{fulkosvog19}.

\begin{example} \label{ex:commutative-SB-not-crisp-determinizable} \rm
Let $\Sigma = \{\gamma^{(1)}, \alpha^{(0)}\}$. We consider the $(\Sigma,\mathrm{ASR}_\N)$-wta $\A=(\{q\},\delta,F)$ with $\delta_0(\varepsilon,\alpha,q)=F_q=0$ and $\delta_1(q,\gamma,q)=1$. Clearly, $\A$ is deterministic and not crisp-deterministic (because $1$ is not one of the unit elements of $\mathrm{ASR}_\N$). Moreover, $\sem{\A}(\gamma^n(\alpha))=n$ for each $n \in \N$. Hence, $\im(\sem{\A})$ is infinite.  \qed
\end{example}
Next we recall three results which we will need in this paper.
The first result is a straightforward generalization of \cite[Lm.~3]{bormalsestepvog06} from semirings to strong bimonoids, cf. also \cite[Thm.~3.9]{fulvog09}.

\begin{lemma}\label{lm:closure-homomorphism}\rm 
    Let $B$ and $C$ be strong bimonoids, $\A$ be $(\Sigma,B)$-wta, and $h: B \to C$ a strong bimonoid homomorphism. We can construct a $(\Sigma,C)$-wta $h(\A)$ such that $\sem{h(\A)}=h \circ \sem{\A}$.
  \end{lemma}
  
\begin{proof} 
Let $\A=(Q,\delta,F)$ be a $(\Sigma,B)$-wta. We introduce the $(\Sigma,C)$-wta $h(\A)=(Q,\delta',F')$ by defining $\delta' = (\delta'_k \mid k \in \N)$ with $\delta'_k= h \circ \delta_k$ for each $k \in \N$
and $F' = h \circ F$. Then it is easy to show that $\sem{h(\A)}=h \circ \sem{\A}$.
\end{proof}

  The second result characterizes the class of weighted tree languages which can be run-recognized by crisp-deterministic wta.

\begin{lemma} \rm cf. \cite[Lm.~5.3]{fulkosvog19} \label{lm:cdwta_finite_image}
  Let $\A$ be a $(\Sigma,B)$-wta. Then the following statements are equivalent.
  \begin{compactitem}
  \item[(i)] There exists a crisp-deterministic $(\Sigma,B)$-wta $\B$ such that $\sem{\A}=\sem{\B}$.
  \item[(ii)] $\sem{\A}$ is a $(\Sigma,B)$-recognizable step mapping.
  \item[(iii)] $\A$ has the finite-image property and the preimage property.
  \end{compactitem}
\end{lemma}

We remark that if  $\B = (Q,\delta,F)$ is a crisp-deterministic  $(\Sigma,B)$-wta,
then for each  $b \in B$  we can construct effectively a finite-state $\Sigma$-tree automaton
which recognizes  $\sem{\B}^{-1}(b)$. Indeed, if  $b \not\in \im(F)$, then  $\sem{\B}^{-1}(b) = \emptyset$.
Therefore let now  $b \in \im(F)$. From $\B$  we immediately obtain a crisp-deterministic
$(\Sigma,\mathbb{B})$-wta $\B_b$  with the same state set and same transitions with non-zero weight
such that the states of  $\B_b$  have final weight $1 \in \mathbb{B}$  iff they have final weight  $b$ in $\B$.
Then $\B_b$  recognizes the tree language $\sem{\B}^{-1}(b)$. This proves our remark.

The third result shows that each $(\Sigma,B)$-wta is crisp-determinizable if $B$ is bi-locally finite. Formally, a $(\Sigma,B)$-wta $\A$ is {\em crisp-determinizable (with respect to the run semantics)} if there exists a crisp-deterministic $(\Sigma,B)$-wta ${\cal C}$ such that $\sem{\A} = \sem{\cal C}$. Thus, $\A$ is crisp-determinizable if $\A$ satisfies the conditions of Lemma~\ref{lm:cdwta_finite_image}.

\begin{lemma} \rm \cite[Cor.~7.5]{fulkosvog19}  \label{lm:bi-loc-ensures-cdwta} 
  Let $\A$ be a $(\Sigma,B)$-wta. If $B$ is bi-locally finite, then $\A$ is crisp-determinizable.
\end{lemma}

\begin{quote}
\em In the rest of this paper, we let $\A = (Q,\delta,F)$ be an arbitrary $(\Sigma,B)$-wta.
\end{quote}

\section{Trim wta}\label{sect:trim-wta}

In this section we define the concept of trim wta. Moreover, we show that, for each wta which satisfies certain simple properties, an equivalent trim wta can be constructed effectively. For this, to each $(\Sigma,B)$-wta $\A$, we associate a context-free grammar 
$\G(\A)$ and show that $\A$ is trim if and only if $\G(\A)$ is reduced. Then we exploit  the fact that for each context-free grammar one can construct effectively an equivalent reduced context-free grammar \cite[Thm.~3.2.3]{har78}.

A state $p \in Q$  is \emph{useful (in $\A$)} if there exist $\xi \in \T_\Sigma$ and  $\rho \in \R_\A(\xi)$ such that  $F_{\rho(\varepsilon)} \ne \0$ and $p \in \im (\rho)$.
The $(\Sigma,B)$-wta $\A$ is \emph{trim} if each of its states is useful.

Let $G=(N,\Delta,P,S)$  be a context-free grammar \cite{har78,hopmotull07}, with nonterminal set $N$, terminal set $\Delta$,  set $P$ of rules, and initial nonterminal $S \in N$. We  denote by $\mathrm{L}(G)$ the language generated by $G$.
A nonterminal $A \in N$ is \emph{useful (in $G$)} if there exist $\alpha,\beta \in (N\cup\Delta)^*$ and $w \in \Delta^*$ such that $S \Rightarrow^* \alpha A \beta \Rightarrow^* w$. Then $G$ is \emph{reduced} if each of its nonterminals is useful \cite[p.~78]{har78}.

To each  $(\Sigma,B)$-wta $\A$, we associate the context-free grammar $\G(\A) = (N,\Delta,P,S)$ where $S$ is a new symbol, $N=Q\cup\{S\}$, $\Delta = \Sigma \cup \Xi$ and $\Xi$ consists of the two parentheses ( and ) and the comma, 
and $P$ is defined as follows:
\begin{compactitem}
\item for each $q \in Q$, if $F_q \ne \0$, then $S \rightarrow q$ is in $P$ and
\item for every $k \in \N$, $\sigma \in \Sigma^{(k)}$, $q_1,\ldots,q_k,q \in Q$:
if  $\delta_k(q_1 \cdots q_k, \sigma,q) \ne \0$, then $q \rightarrow \sigma(q_1, \ldots, q_k)$ is in $P$.
\end{compactitem}
Then it can be shown that, for each $\xi \in \T_\Sigma$, there exists a bijection between the set of runs $\rho\in \R_\A(\xi)$ with $F_{\rho(\varepsilon)}\ne \0$
and the set of leftmost derivations of  $\G(\A)$ for $\xi$, cf. Figure \ref{fig:run-rule-tree}.
Hence a state $p\in Q$ is useful in $\A$ if and only if it is useful in  $\G(\A)$. This implies that $\A$ is trim if and only if $\G(\A)$ is reduced.

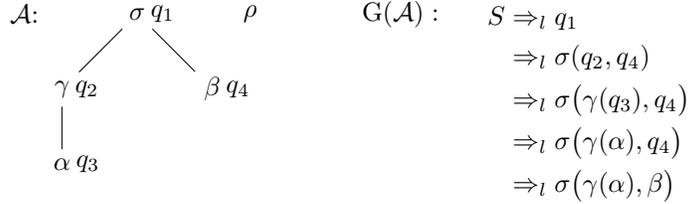
\begin{figure}[t]
\centering
\begin{tikzpicture}
\node at (-1.5,0) {$\A$:};
\node at (1.5,0) {$\rho$};
\node at (0,0) (eps) {$\sigma$} node[right = -5pt of eps] {$q_1$};
\node at (-1,-1) (1) {$\gamma$} node[right = -5pt of 1] {$q_2$};
\node at (1,-1) (2) {$\beta$} node[right = -5pt of 2] {$q_4$};
\node at (-1,-2) (11) {$\alpha$} node[right = -5pt of 11] {$q_3$};

\draw[-]
  (eps) -- (1)
  (eps) -- (2)
  (1) -- (11)
;

\node[anchor=north] at (6,0.75) {
\begin{minipage}{2em}
\begin{align*}
S &\Rightarrow_l q_1\\
&\Rightarrow_l \sigma(q_2,q_4)\\
&\Rightarrow_l \sigma\big(\gamma(q_3),q_4\big)\\
&\Rightarrow_l \sigma\big(\gamma(\alpha),q_4\big)\\
&\Rightarrow_l \sigma\big(\gamma(\alpha),\beta\big)
\end{align*}
\end{minipage}
};
\node at (3.5,0) {$\G(\A):$};

\end{tikzpicture}
\caption{\label{fig:run-rule-tree} A run $\rho \in \R_\A(q_1,\xi)$, where $F_{q_1}\ne \0$ and $\xi = \sigma(\gamma(\alpha),\beta)$ and the leftmost derivation of $\G(\A)$ for $\xi$ which corresponds to $\rho$.}
\end{figure}

We say that the strong bimonoid  $B$ \emph{has an effective test for $\0$}  if for each  $b \in B$  we can decide whether $b = \0$.

\begin{theorem} \label{thm:A'-equivalent-to-A} \cite[Lm.~5]{drofulkosvog20}
  Let $B$ have an effective test for $\0$ and $\A$ be a $(\Sigma,B)$-wta. If $\A$ is given effectively and has a useful state,
  then we can construct effectively a $(\Sigma,B)$-wta $\A'$  such that  $\A'$ is trim and $\sem{\A'}=\sem{\A}$. If $\A$ is finitely ambiguous, then $\A'$ is so.
\end{theorem}
\begin{proof}
  Using the effective test for $\0$, we can construct effectively the context-free grammar $\G(\A)=(N,\Delta,P,S)$ (recall that $N = Q \cup \{S\}$). Due to our assumption on $\A$ we have $\mathrm{L}(\G(\A))\ne \emptyset$.  Thus, by \cite[Thm.~3.2.3]{har78}, a context-free grammar $G'=(N',\Delta,P',S)$ can be constructed effectively such that $G'$ is reduced and $\mathrm{L}(G')= \mathrm{L}(\G(\A))$. By the proof of that theorem, we know that $N' = Q' \cup \{S\}$, where $Q'$ is the set of all useful nonterminals in $Q$. Hence $Q'$ is the set of all useful states of $\A$. Moreover, $Q'\ne \emptyset$ by our assumption on $\A$.

Now let $\A' = (Q',\delta',F')$ be the $(\Sigma,B)$-wta such that, for each $k\in \N$, $\delta'_k=\delta_k|_{(Q')^k\times \Sigma^{(k)}\times Q'}$, and $F'=F|_{Q'}$.
It is obvious that $\A'$ is trim and $\A'$ is finitely ambiguous if $\A$ is so.

Lastly we prove that $\sem{\A} = \sem{\A'}$. Let $\xi \in \T_\Sigma$. Obviously, $\R_{\A'}(\xi) \subseteq \R_\A(\xi)$ and for each $\rho \in \R_{\A'}(\xi)$ we have $\wt_{\A'}(\xi,\rho) = \wt_\A(\xi,\rho)$. If $\rho \in \R_\A(\xi)\setminus \R_{\A'}(\xi)$, then there exist $p\in \im(\rho)$ such that $p$ is not useful. Then $F_{\rho(\varepsilon)}=\0$ and hence $\wt_\A(\xi,\rho) \otimes F_{\rho(\varepsilon)}=\0$.  
  Thus we can compute
  \begin{align*}
    \sem{\A}(\xi) = \bigoplus_{\rho \in \R_\A(\xi)} \wt_\A(\xi,\rho) \otimes F_{\rho(\varepsilon)}
                  =  \bigoplus_{\rho \in \R_{\A'}(\xi)} \wt_{\A'}(\xi,\rho) \otimes F'_{\rho(\varepsilon)}
    =\sem{\A'}(\xi).
    \end{align*}
  \end{proof}

\section{Pumping lemma} \label{sect:pumping-lemma}

In this section, we wish to prove a pumping lemma for runs of weighted tree automata. Pumping lemmas are used in order to achieve structural implications on small or particular large trees (cf. \cite[Lm.~2.10.1]{gecste84} and \cite[Lm.~5.5]{bor04}).
Essentially, we follow the classical approach for unweighted tree automata combined with an analysis
of Equality~\eqref{eq:weight-of-a-run}. Assume we are given a wta $\A$  with state set $Q$, a tree  $\xi \in \T_\Sigma$
with height greater than $|Q|$  and a run  $\kappa$  of  $\A$  on  $\xi$. As for unweighted tree automata,
choose a path, i.e., a linearly ordered subset of positions, in $\xi$  whose length equals  $\hgt(\xi)$.
Clearly, there are two positions  $u,v$  in this path with  $\kappa(u) = \kappa(v)$  in  $Q$; say, $u$  is  above $v$, i.e., there exists $w \in \bbN^*_+$ such that $v=uw$ (cf. Figure~\ref{fig:viz-pumping-intro}).
Now we consider the subtrees  $\xi|_u$ (respectively, $\xi|_v$)  comprising all positions of $\xi$  which are equal to or below $u$ (respectively, $v$). Clearly, we can cut out the context $(\xi|_u)|^w$, thus replacing the subtree $\xi|_u$ by  $\xi|_v$  and obtaining a smaller tree for which a restriction of the run $\kappa$ leads to the same state as $\kappa$. But, we can also substitute a copy of the context $(\xi|_u)|^w$ at position  $v$. We obtain a tree  $\xi'$, and we can copy the corresponding part of the mapping $\kappa$ to obtain a run $\kappa'$  on  $\xi'$ leading again to the same final state as $\kappa$.

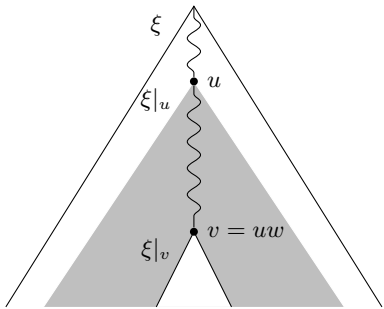
\begin{figure}[t]
	\centering
	\small
	\begin{tikzpicture}
	    \draw (0,0) -- (2.5,4) -- (5,0); 
	    \draw[fill=lightgray,draw=white] (0.5,0) -- (2.5,3) -- (4.5,0);
	    \draw[fill=white] (2,0) -- (2.5,1) -- (3,0);
	    \node at (2,3.75) {$\xi$};
	    \node at (2,2.75) {$\xi|_u$};
	    \node at (2,0.75) {$\xi|_v$};
	    \node[fill=black,circle,inner sep=0pt,minimum size=3pt,label=right:{$u$}] at (2.5,3) (u) {};
	    \node[fill=black,circle,inner sep=0pt,minimum size=3pt,label=right:{$v=uw$}] at (2.5,1) (v) {};
	    \draw[decorate, decoration={snake,pre length=4pt, post length=1pt}] (2.5,4) -- (u);
	    \draw[decorate, decoration={snake,pre length=3pt, post length=3pt}] (u) -- (v);
	\end{tikzpicture}
	\caption{\label{fig:viz-pumping-intro} Illustration of the tree $\xi$, the positions $u$ and $v$, and the subtrees $\xi|_u$ and $\xi|_v$. Moreover, the context  $(\xi|_u)|^w$ is shaded.}
\end{figure}

Now when we calculate the weight of $\kappa'$  on  $\xi'$  according to Equality~\eqref{eq:weight-of-a-run}, a careful analysis shows
that after an insertion process two factors of the product of weights originating from transitions at positions of the context $(\xi|_u)|^w$ get repeated. Hence, if we repeat the insertion process,
then the two factors are replaced by their powers.

To make this outline formally exact, we show that the product of weights of the run $\kappa$  on the context $(\xi|_u)|^w$ splits into two factors, a 'left one' and  'right one' (see Observation~\ref{obs:decomp-run-left-right}). This product of weights with the insertion process performed once is analyzed in Lemma~\ref{lm:combining-runs}; it turns out that splitting described before leads to
an additional 'left factor' and an additional 'right factor'. The consequence of the cutting respectively replacing process
in general for the product of weights is described in Lemma~\ref{lm:decomposition-of-a-run}; the multiple insertion process leads to
an additional 'left power' and an additional 'right power'. The whole pumping lemma is given in Theorem~\ref{thm:pumping-lemma-for-runs}. 

The question may arise why we present another pumping lemma and do not use an existing one (cf. \cite[Lm.~5.5]{bor04}). To answer this we note that Borchardt's setting deals with deterministic wta over semirings and employs initial algebra semantics, whereas in Theorem~\ref{thm:pumping-lemma-for-runs} we deal with (arbitrary) wta over strong bimonoids and employ run semantics. Nevertheless, if we consider the class of all deterministic wta over semirings, then the two settings coincide.

Now we introduce some notations.

Let $c \in \C_\Sigma$,
$\zeta \in \T_\Sigma$, $\{v\} = \pos_\square(c)$, $q',q \in Q$,
$\rho \in \R_\A(q',c,q)$, and $\theta \in \R_\A(q,\zeta)$.
The {\em combination of $\rho$ and $\theta$}, denoted by $\rho[\theta]$, is
the mapping $\rho[\theta]: \pos(c[\zeta]) \to Q$ defined for every $u \in \pos(c[\zeta])$
as follows: if $u=vw$ for some $w$, then we define $\rho[\theta](u)=\theta(w)$,
otherwise we define $\rho[\theta](u)=\rho(u)$. Clearly, $\rho[\theta] \in \R_\A(q',c[\zeta])$.
For every $\xi \in \T_\Sigma$, $\rho \in \R_\A(\xi)$, and $v \in \pos(\xi)$, we define the
run $\rho|^v$ on the context $\xi|^v$ such that for every $w \in \pos(\xi|^v)$ we set
$\rho|^v(w) =\rho(w)$. If $\rho \in \R_\A(\square)$ with $\rho(\varepsilon)=q$ for some $q \in Q$,
then sometimes we write $\widetilde{q}$ for $\rho$.

Let $c \in \C_\Sigma$, $\{v\} =\pos_\square(c)$, and $\rho \in \R_\A(c)$. We define two mappings $l_{c,\rho}: \prefix(v) \to B$ and $r_{c,\rho}: \prefix(v) \to B$ inductively on the length of their arguments (cf. \cite[p.~526]{bor04} for deterministic wta). Intuitively, the product \eqref{eq:weight-of-a-run}, which yields the element $\wt(c,\rho) \in B$, can be split into a left subproduct $l_{c,\rho}(\varepsilon)$ and a right subproduct $r_{c,\rho}(\varepsilon)$, where the border is given by the factor $\mathbb{1}$ coming from the weight of $\Box$. Figure \ref{fig:l-and-r} shows the illustration of mappings $l_{c,\rho}$ and $r_{c,\rho}$. Formally, let $w \in \prefix(v)$. Then, assuming that $c(w)=\sigma$ and $\rk(\sigma)=k$, we let
\[
  l_{c,\rho}(w) =
  \begin{cases}
 \1   & \text{ if $w=v$}\\
 \bigotimes_{j=1}^{i-1} \wt(c|_{wj},\rho|_{wj}) \otimes l_{c,\rho}(wi)  & \text{ if $wi \in \prefix(v)$ for some $i \in \N_+$}
    \end{cases}
  \]
\[   r_{c,\rho}(w) =
  \begin{cases}
 \1  & \hspace*{-14mm}\text{ if $w=v$}\\
 r_{c,\rho}(wi) \otimes \bigotimes_{j=i+1}^{k} \wt(c|_{wj},\rho|_{wj}) \otimes \delta_k(\rho(w1) \cdots \rho(wk),\sigma,\rho(w)) &\\
 &\hspace*{-54mm} \text{ if $wi \in \prefix(v)$ for some $i \in \N_+$}\enspace.
    \end{cases}
  \] 
In the sequel, we abbreviate $l_{c,\rho}(\varepsilon)$ and $r_{c,\rho}(\varepsilon)$ by $l_{c,\rho}$ and $r_{c,\rho}$, respectively.

\begin{sidewaysfigure}
\centering
\begin{tikzpicture}
\footnotesize
\node[anchor=north] at (-7.1,1.5) (t1) {$\wt(c|_{w1},\rho|_{w1})$};
\node[anchor=north] at (-7.1,0.75) {$c|_{w1}$};

\node at (-5.6,1.25) {$\otimes \ldots \otimes$};

\node[anchor=north] at (-3.625,1.5) (ti-1) {$\wt(c|_{w(i-1)},\rho|_{w(i-1)})$};
\node[anchor=north] at (-3.625,0.75) {$c|_{w(i-1)}$};

\node at (-2, 1.25) {$\otimes$};
\node[anchor=north] at (-1.1,1.5) {$l_{c,\rho}(wi)$};
\draw[rounded corners=7pt] (-0.475,1) rectangle (-1.725,1.5); 

\node[anchor=north] at (0,0.75) {$c|_{wi}$};
\draw (0,0) -- (0,-0.25) node [midway,right] {$v$};
\node at (0,-0.375) {$\square$};

\node at (2,1.25) {$\otimes$};
\node[anchor=north] at (1.1,1.5) {$r_{c,\rho}(wi)$};
\draw[rounded corners=7pt] (0.475,1) rectangle (1.725,1.5); 

\node[anchor=north] at (3.625,1.5) (ti+1) {$\wt(c|_{w(i+1)},\rho|_{w(i+1)})$};
\node[anchor=north] at (3.625,0.75) {$c|_{w(i+1)}$};

\node at (5.6,1.25) {$\otimes \ldots \otimes$};

\node[anchor=north] at (7.1,1.5) (tk) {$\wt(c|_{wk},\rho|_{wk})$};
\node[anchor=north] at (7.1,0.75) {$c|_{wk}$};

\node at (8.2,1.25) {$\otimes$};

\node at (0,4.5) (s) {$\sigma$};
\node at (-1.125,4.5) {$l_{c,\rho}(w)$};
\draw[rounded corners=7pt] (-0.5,4.25) rectangle (-1.75,4.75); 
\node at (1.125,4.5) {$r_{c,\rho}(w)$};
\draw[rounded corners=7pt] (0.5,4.25) rectangle (1.75,4.75);

\node at (-7.1,1.75) (da11) {};
\node at (-2,4.5) (da12) {};

\draw[-Implies,double distance=2pt] (da11) to [out=90,in=180] (da12);

\node at (6.5,3.25) (da21) {};
\node at (2,4.5) (da22) {};

\draw[-Implies,double distance=2pt] (da21) to [out=90,in=0] (da22);

\draw[-]
  (-8.1,0) -- (-6.1,0) -- (t1) -- (-8.1,0)
  (-2.625,0) -- (-4.625,0) -- (ti-1) -- (-2.625,0)
  (-1,0) -- (1,0) -- (0,1.5) -- (-1,0)
  (2.625,0) -- (4.625,0) -- (ti+1) -- (2.625,0)
  (6.1,0) -- (8.1,0) -- (tk) -- (6.1,0)
  (s) -- (t1) node[midway,xshift=-24pt,anchor=north] {$1$}
  (s) -- (ti-1) node[midway,xshift=-20pt,anchor=north] {$i-1$}
  (s) -- (0,1.5) node[midway,xshift=-8pt,anchor=north] {$i$}
  (s) -- (ti+1) node[midway,xshift=-8pt,anchor=north] {$i+1$}
  (s) -- (tk) node[midway,xshift=-6pt,anchor=north] {$k$}
;

\node at (-4,2) {$\cdots$};
\node at (4,2) {$\cdots$};

\node at (6.5,3) {$\delta_k(\rho(w1) \cdots \rho(wk),\sigma,\rho(w))$};

\draw[dashed,smooth,line cap=round, rounded corners=10pt] (-8.1,0.875) -- (-0.375,0.875) -- (-0.375,1.625) -- (-8.1,1.625) -- cycle;

\draw[dashed,smooth,line cap=round, rounded corners=10pt] (8.5,0.875) -- (0.375,0.875) -- (0.375,1.625) -- (4.5,1.625) -- (4.5,3.25) -- (8.5,3.25) -- cycle;

\draw[decorate, decoration=snake] (0,6.5) -- (s);
\node[anchor=west] at (0,5.5) {$w$};
\draw (0,6.5) -- (-7.75,4.5);
\node at (-4,5.75) {$c \in \C_\Sigma$};
\draw (0,6.5) -- (8.25,4.5);
\node at (4.5,5.75) {$\rho \in \R_\A(c)$};

\end{tikzpicture}
\caption{\label{fig:l-and-r} Illustration of mappings $l_{c,\rho}$ and $r_{c,\rho}$}
\end{sidewaysfigure}
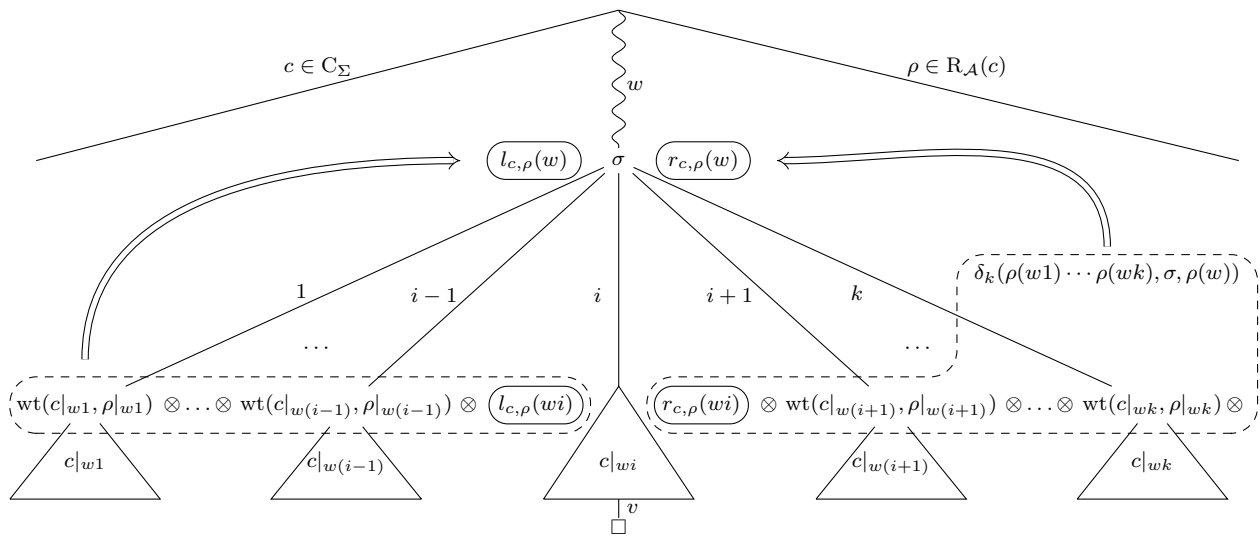

\begin{observation} \label{obs:decomp-run-left-right} \rm
Let $c \in \C_\Sigma$ and  $\rho \in \R_\A(c)$. Then $\wt(c,\rho) = l_{c,\rho} \otimes r_{c,\rho}$.
\end{observation}

The next lemma can be proved by an easy and straightforward induction on~$c$.
 
\begin{lemma} \rm \label{lm:combining-runs} (cf. \cite[Lm.~5.1]{bor04})
Let $c \in \C_\Sigma$, $\zeta \in \T_\Sigma$, $q', q \in Q$, $\rho \in \R_\A(q',c,q)$, and $\theta \in \R_\A(q,\zeta)$. Then $\wt(c[\zeta],\rho[\theta]) = l_{c,\rho} \otimes \wt(\zeta,\theta) \otimes  r_{c,\rho}$.
\end{lemma}

Let $c \in \C_\Sigma$, $q \in Q$, and $\rho \in \R_\A(q,c,q)$ be a loop. For each $n \in \N$, the {\em nth power of $\rho$}, denoted by $\rho^n$, is the run on $c^n$ defined by induction as follows: $\rho^0=\sqrun{q}$ (note that $c^0=\square$) and $\rho^{n+1}=\rho[\rho^n]$. Next we apply the previous results to the weights of powers of loops.

\begin{lemma}\rm \label{lm:decomposition-of-a-run} (cf. \cite[Lm.~5.3]{bor04})
Let $c',c \in \C_\Sigma$ and $\zeta \in \T_\Sigma$, $q',q \in Q$, $\rho' \in \R_\A(q',c',q)$, $\rho \in \R_\A(q,c,q)$, and $\theta \in \R_\A(q,\zeta)$. Then, for each $n \in \N$,  
\[
  \wt(c'\big[c^n[\zeta]\big], \rho'\big[\rho^n[\theta]\big]) =
  l_{c',\rho'} \otimes (l_{c,\rho})^n \otimes \wt(\zeta,\theta) \otimes (r_{c,\rho})^n \otimes r_{c',\rho'}
  \enspace.
\]
\end{lemma}

\begin{proof} We can prove easily by induction on $n$ that
  \begin{equation}
\wt(c^n[\zeta],\rho^n[\theta]) = (l_{c,\rho})^n \otimes \wt(\zeta,\theta) \otimes  (r_{c,\rho})^n \text{  for each $n \in \N$}\enspace. \label{eq:n-context}
    \end{equation}

Then for each $n\in N$ we have
  \begin{align*}
    \wt(c'\big[c^n[\zeta]\big], \rho'\big[\rho^n[\theta]\big])
    &= l_{c',\rho'} \otimes \wt(c^n[\zeta],\rho^n[\theta]) \otimes r_{c',\rho'} \tag{\text{by Lemma \ref{lm:combining-runs}}}\\
    &= l_{c',\rho'} \otimes (l_{c,\rho})^n \otimes \wt(\zeta,\theta) \otimes (r_{c,\rho})^n \otimes r_{c',\rho'} \tag{\text{by \eqref{eq:n-context}}} 
  \end{align*} 
\end{proof}

Finally, we recall from \cite{bor04} the pumping lemma for runs of $\A$ on trees in $\T_\Sigma$ which are large enough. We note that $B$ need not be commutative.

\begin{theorem}  \label{thm:pumping-lemma-for-runs} (pumping lemma, cf. \cite[Lm.~5.5]{bor04})
Let $\xi \in \T_\Sigma$, $q' \in Q$, $\kappa \in \R_\A(q',\xi)$. If $\hgt(\xi) \geq |Q|$, then
there exist  $c',c \in \C_\Sigma$, $\zeta \in \T_\Sigma$, $q \in Q$, $\rho' \in \R_\A(q',c',q)$, $\rho \in \R_\A(q,c,q)$, and $\theta \in \R_\A(q,\zeta)$ such that $\xi=c'\big[c[\zeta]\big]$, $\kappa=\rho'\big[\rho[\theta]\big]$,
$\hgt(c) > 0$, $\hgt\big(c[\zeta]\big) < |Q|$, and, for each $n \in \N$,
\[
  \wt(c'\big[c^n[\zeta]\big], \rho'\big[\rho^n[\theta]\big]) =
  l_{c',\rho'} \otimes (l_{c,\rho})^n \otimes \wt(\zeta,\theta) \otimes (r_{c,\rho})^n \otimes r_{c',\rho'}
  \enspace.
\]
\end{theorem}

\begin{proof}
  Since $\hgt(\xi) \geq |Q|$ there exist  $u,w \in \N_+^*$ such that $uw \in \pos(\xi)$, $|w| > 0$, $\hgt(\xi|_u) < |Q|$, and $\kappa(u)=\kappa(uw)$. Then we let $c'=\xi|^u$, $c=(\xi|_u)|^w$, $\zeta=\xi|_{uw}$.  Clearly, \(\xi=c'\big[c[\zeta]\big]\).
  Moreover, we set $\rho'=\kappa|^u$, $\rho=(\kappa|_u)|^w$ and $\theta=\kappa|_{uw}$.
Then the statement follows from Lemma~\ref{lm:decomposition-of-a-run}. 
\end{proof}

We say that \emph{small loops of $\A$ have weight~$\1$} if, for every $q \in Q$, $c \in \C_\Sigma$, and loop $\rho \in \R_\A(q,c,q)$, if $\hgt(c) < |Q|$, then $\wt(c,\rho) = \1$.

\begin{lemma} \label{lm:weight-of-reduced-tree} \rm
  Let $B$ be commutative or one-product free. If small loops of $\A$ have weight $\1$, then, for every $\xi \in \T_\Sigma$, $q' \in Q$, and $\kappa \in \R_\A(q',\xi)$, there exist $\xi' \in \T_\Sigma$ and $\kappa' \in \R_\A(q',\xi')$ such that $\hgt(\xi') < |Q|$ and $\wt(\xi,\kappa) = \wt(\xi', \kappa')$. 
  \end{lemma}
  
  \begin{proof}    Let $\xi \in \T_\Sigma$, $q' \in Q$, and $\kappa \in \R_\A(q',\xi)$. We may assume that $\hgt(\xi) \geq |Q|$.
  Applying Theorem~\ref{thm:pumping-lemma-for-runs} (for $n=1$ and $n=0$),
  there exist  $c,c' \in \C_\Sigma$, $\zeta \in \T_\Sigma$, $q \in Q$, $\rho' \in \R_\A(q',c',q)$, $\rho \in \R_\A(q,c,q)$, and $\theta \in \R_\A(q,\zeta)$ 
  such that   $\xi=c'\big[c[\zeta]\big]$, $\kappa = \rho'\big[\rho[\theta]\big]$,  $\hgt(c) > 0$, $\hgt\big(c[\zeta]\big) < |Q|$, and
  \begin{align*}
  \wt(\xi,\kappa) = & \wt(c'\big[c[\zeta]\big], \rho'\big[\rho[\theta]\big]) =
    l_{c',\rho'} \otimes  l_{c,\rho} \otimes \wt(\zeta,\theta) \otimes  r_{c,\rho} \otimes r_{c',\rho'}\enspace,\\
  & \wt(c'[\zeta],\rho'[\theta]) = l_{c',\rho'} \otimes  \wt(\zeta,\theta) \otimes   r_{c',\rho'}\enspace.
  \end{align*}
  By our assumption $\wt(c,\rho) = \1$, and by Observation \ref{obs:decomp-run-left-right} we have 
  $\wt(c,\rho) = l_{c,\rho} \otimes r_{c,\rho}$. If $B$ is commutative, then $\wt(\xi,\kappa) = \wt(c'[\zeta],\rho'[\theta])\otimes l_{c,\rho} \otimes r_{c,\rho}$.
If $B$ is one-product free, then $l_{c,\rho} =  r_{c,\rho}= \1$. Hence in both cases we have $\wt(\xi,\kappa)=\wt(c'[\zeta],\rho'[\theta])$.
  
  Note that  $\rho'[\theta] \in \R_\A(q',c'[\zeta])$  and  $\size(c'[\zeta]) < \size (\xi)$. If $\hgt(c'[\zeta])<|Q|$, then we are ready. 
  Otherwise we continue with  $c'[\zeta]$, $q'$, and  $\rho'[\theta]$ as before.
  After finitely many steps, we obtain  $\xi' \in T_\Sigma$  and $\kappa' \in \R_\A(q',\xi')$
  with  $\hgt(\xi') < |Q|$  as required. 
  \end{proof}

\section{The preimage property and a sufficient condition for a wta to be crisp-determinizable}\label{sect:sufficient-condition}

An important result for recognizable weighted string languages $r$ over the semiring $\N$ is that, for each $n \in \N$, the string language $r^{-1}(n)$ is recognizable (preimage property, \cite[III. Cor. 2.5]{berreu88}). In this section we show some variants of the preimage property. Then, as a main result of this section, we give a sufficient condition for a wta $\A$ over an arbitrary strong bimonoid $B$ which guarantees that $\sem{\A}$ is a recognizable step mapping (cf. Theorem~\ref{thm:sufficient-conds-ensure-rec-step-map}) and thus, in particular, $\sem{\A}$ satisfies the preimage property and $\A$ is crisp-determinizable by Lemma~\ref{lm:cdwta_finite_image}. We also show that if  the strong bimonoid is computable and $\A$ is given effectively, then the crisp-deterministic wta equivalent to $\A$ can be constructed effectively. Finally, as an application of Theorem~\ref{thm:sufficient-conds-ensure-rec-step-map}, we extend the mentioned preimage property \cite[III. Cor. 2.5]{berreu88} to wta over past-finite monotonic strong bimonoids (cf. Theorem~\ref{thm:inverse_b_recognizable}).

We say that $(B,\oplus,\otimes,\0,\1)$ is {\em computable} if $B$ is a recursively enumerable set with tests for equality and the operations $\oplus$ and $\otimes$ are computable (e.g., by a Turing machine).

Subsequently, we will need the following concepts. We define the sets 
\begin{align*}
  \rH(\A) &= \{\wt(\xi,\rho) \mid \xi \in \T_\Sigma \text{ and } \rho \in \R_\A(\xi)\}\text{ and }\\  
  \rC(\A) &=\{\wt(\xi,\rho) \otimes F_{\rho(\varepsilon)} \mid \xi \in \T_\Sigma, \rho \in \R_\A(\xi)\}\enspace.	
\end{align*}
 We call elements of $\rC(\A)$  {\em complete run weights of $\A$}. Observe that if $\rH(\A)$ is a finite set, then $\rC(\A)$ is also finite because $\rC(\A)\subseteq \rH(\A)\otimes \im(F)$.
The following notions will be needed for our main result of this section. 

Let $b\in B$. If $b$ has finite additive order, then there exists a least number $i \in \N_+$ such that $ib=(i+k)b$ for some $k \in \N_+$, and there exists a least number $p \in \N_+$ such that $ib = (i+p)b$. We call $i$ the {\em index (of $b$)} and $p$ the {\em period (of $b$)}, and denote them by $i(b)$ and $p(b)$, respectively. Moreover, we call $i+p-1$, i.e., the number of elements of $\langle b \rangle_\oplus$, the {\em order of $b$}. 

Then for each $b\in \rC(\A)$, we define the mapping $f_{\A,b}: \T_\Sigma \to \N$, called  the {\em complete run number mapping of $b$}, by
\[f_{\A,b}(\xi) = |\{\rho \in \R_\A(\xi) \mid \wt(\xi,\rho) \otimes F_{\rho(\varepsilon)} = b\}|\]
for each $\xi\in\T_\Sigma$. The mapping $f_{\A,b}$ is {\em bounded}, if there exists $K \in \N$ such that $f_{\A,b}(\xi) \leq K$ for each $\xi \in \T_\Sigma$. 
Clearly, if $\A$ is finitely ambiguous, then there exists $K \in \bbN$ such that,  for each $\xi \in \rT_\Sigma$, we have  $|\rR_\cA^{F \neq \bb0}(\xi)| \leq K$, and thus, $f_{\A,b}$ is bounded by $K$ for each $b\in \rC(\A)\setminus\{\0\}$.

\begin{theorem} \label{thm:sufficient-conds-ensure-image-finite}
   Let $\A=(Q,\delta,F)$ be a $(\Sigma,B)$-wta such that $\rH(\A)$ is finite. If, for each $b \in \rC(\A)$, the mapping $f_{\A,b}$ is bounded or $b$ has finite additive order, then $\sem{\A}$ has the finite-image property. 
\end{theorem}
    
\begin{proof} We note that $\rC(\A)$ is finite because $\rH(\A)$ is finite. Then for each $\xi \in \T_\Sigma$, we have
\begin{equation}\label{eq:semantics=sum-f}
\sem{\A}(\xi)=\bigoplus_{b\in \rC(\A)} \big(f_{\A,b}(\xi)\big)b\enspace.
\end{equation}
Let $b\in \rC(\A)$.
If the mapping $f_{\A,b}$ is bounded by $K$ for some $K\in \N$, then $\big(f_{\A,b}(\xi)\big)b \in \{jb \mid j \in [0,K]\}$.
Otherwise $b$ has finite additive order. Thus we have $\big(f_{\A,b}(\xi)\big)b \in \{0, b, 2b, \ldots, (i(b)+p(b)-1)b\}$. 
\end{proof}

Next, we wish to show that under the assumptions of Theorem~\ref{thm:sufficient-conds-ensure-image-finite},
$\A$  also has the preimage property. For the main result of this section,
we will need the following preparation.

\begin{lemma} \rm \label{lm:inverse_b_recognizable_with_B_finite_semiring}
     Let $B$ be a finite semiring and $\A$ be a $(\Sigma,B)$-wta. Then the following statements hold.
    \begin{compactenum}
        \item $\A$ has the preimage property.
        \item If $B$ is computable and $\A$ is given effectively, then, for each $b \in B$, we can construct effectively a finite-state $\Sigma$-tree automaton which recognizes $\sem{\A}^{-1}(b)$.
    \end{compactenum}
\end{lemma}

\begin{proof}
  If, in addition, $B$ is commutative, then Statement 1 immediately follows from \cite[Lm.~6.1]{drovog06} and Statement 2 is also clear by the proof of that lemma. However, we can drop the condition that $B$ is commutative because of the following. In \cite{drovog06}, for each $\xi \in \T_\Sigma$, the weight of a run $\rho \in \R_\A(\xi)$  is defined as the product of the weight of the transitions determined by the run: 
  $\wt(\xi,\rho)=\bigotimes_{w \in\pos(\xi)}\delta_k(\rho(w1)\cdots\rho(wk),\xi(w),\rho(w))$, where $\xi(w)\in \Sigma^{(k)}$ and the factors are multiplied in an arbitrary order. However,  in the proof of \cite[Lm.~6.1]{drovog06}, we observe that the order of the factors in such products does not change. Hence that proof is also valid for an arbitrary, fixed order of elements in the products. In particular, it is valid for the depth-first left-to-right order with which we have defined $\wt(\xi,\rho)$ (cf. \eqref{eq:weight-of-a-run}).  Hence the proof of \cite[Lm.~6.1]{drovog06} is also valid for  our setting.
\end{proof}

\begin{lemma} \rm\label{cor:inverse_n_recognizable}
Let $\A$ be a $(\Sigma,\N)$-wta. Then the following statements hold.
    \begin{compactenum}
        \item $\A$ has the preimage property.
        \item If $\A$ is given effectively, then, for each $n \in \N$, we can construct effectively a finite-state $\Sigma$-tree automaton which recognizes $\sem{\A}^{-1}(n)$.
    \end{compactenum}
\end{lemma}

\begin{proof}
    Statement 1 was proved in \cite[Lm.~6.3(2)]{drovog06}. Now we prove Statement~2. It was stated already in \cite{drovog06}, but we include the proof for the sake of completeness.  Let $n \in \N$ and $M = \{k \in \N \mid n < k\}$. Moreover, let $\sim$ be the equivalence relation on the set $\N$ defined such that its classes are the singleton sets $\{k\}$ for each $k \in [0,n]$ and the set $M$. As is well known, $\sim$  is a congruence, which can be seen as follows.
 Let $k,k' \in M$ and $m \in \N$. Obviously, $k \leq k + m$, and hence $k + m \in M$ and similarly $k' + m \in M$. Moreover, if $m \neq 0$, then we have $k \cdot m \in M$ and $k' \cdot m \in M$. Thus $M$ is a congruence class and the relation $\sim$ is a congruence on $\N$. Then the quotient semiring $\N/_\sim$ is finite. Let $h: \N \to \N/_\sim$ be the canonical semiring homomorphism. Clearly, we can give effectively the congruence classes of $\sim$, i.e., the elements of $\N/_\sim$, by choosing only one representative for each congruence class. Due to this fact and that $\N$ is computable, the semiring $\N/_\sim$ is computable. By Lemma~\ref{lm:closure-homomorphism}, we can construct effectively the $(\Sigma,\N/_\sim)$-wta $h(\A)$ such that $\sem{h(\A)} = h \circ \sem{\A}$. Since $\N/_\sim$ is finite and computable, by Lemma \ref{lm:inverse_b_recognizable_with_B_finite_semiring}(2), we can construct effectively  a finite-state $\Sigma$-tree automaton which recognizes $(h \circ \sem{\A})^{-1}(\{n\})= \sem{\A}^{-1}(n)$. 
\end{proof}

For every $m\in \N$ and $n\in \N_+$, we define $m+n\cdot\N=\{m+n\cdot j \mid j\in \N\}$. Let us denote by $\N/n\N$ the semiring of nonnegative integers modulo $n$. Moreover, for $m \in \N$, we let  $\overline{m} = m+n\N$, the residue class of $m$ modulo $n$.

\begin{lemma} \rm\label{lm:inverse_m+nN_recognizable} (cf. \cite[III. Cor. 2.4]{berreu88})
    Let $\A$ be a $(\Sigma,\N)$-wta.  Then the following statements hold.
    \begin{compactenum}
        \item For every $m\in \N$ and $n\in \N_+$, the $\Sigma$-tree language $\sem{\A}^{-1}(m+n\cdot\N)$ is recognizable.
        \item If $\A$ is given effectively, then, for every $m\in \N$ and $n\in \N_+$, we can construct effectively a finite-state $\Sigma$-tree automaton which recognizes $\sem{\A}^{-1}(m+n\cdot\N)$.
    \end{compactenum}
\end{lemma}

\begin{proof} 
Let us abbreviate $\sem{\A}$ by $r$.

\underline{Proof of 1:} Let $m \in \N$. If $m < n$, then, by Lemma~\ref{lm:closure-homomorphism}, $(h \circ r) \in \mathrm{Rec}(\Sigma,\N/n\N)$, where
$h: \N \to \N/n\N$ is the canonical semiring homomorphism. Moreover, $r^{-1}(m + n \cdot \N) = r^{-1}\big(h^{-1}(\overline{m})\big) = (h \circ r)^{-1}(\overline{m})$. Since $\N/n\N$ is a finite semiring, by Lemma~\ref{lm:inverse_b_recognizable_with_B_finite_semiring}(1), the $\Sigma$-tree language $(h \circ r)^{-1}(\overline{m})$ is recognizable. Now assume that $m \geq n$. Then there exist  $m' \in [0,n-1]$ and $k \in \N_+$ such that $m = m' + n \cdot k$. Then
\[r^{-1}(m + n \cdot \N) = r^{-1}(m' + n \cdot \N) \setminus \bigcup_{j=0}^{k-1} r^{-1}(m' + n \cdot j)\enspace.\]
As we saw, the $\Sigma$-tree language $r^{-1}(m' + n \cdot \N)$ is recognizable because $m'<n$. Moreover, by Lemma~\ref{cor:inverse_n_recognizable}(1), for each $j \in [0,k-1]$, the $\Sigma$-tree language $r^{-1}(m' + n \cdot j)$ is also recognizable. Finally, $\Sigma$-tree languages are closed under union and subtraction. Thus, also in this case, $r^{-1}(m + n \cdot \N)$ is recognizable.

\underline{Proof of 2:} We follow the proof of Statement 1. Let $m \in \N$. Assume that $m < n$. Obviously, we can  give effectively the residue classes modulo $n$, i.e., the elements of $\N/n\N$, by choosing only one representative for each residue class. Because of this fact and that $\N$ is computable, the semiring $\N/n\N$ is also computable. Since $\A$ is  given effectively, by Lemma~\ref{lm:closure-homomorphism}, we can  construct effectively the $(\Sigma,\N/n\N)$-wta $h(\A)$ such that $\sem{h(\A)}=h \circ r$. Since $\N/n\N$ is a  computable finite semiring, by Lemma~\ref{lm:inverse_b_recognizable_with_B_finite_semiring}(2), we can construct effectively a finite-state $\Sigma$-tree automaton which recognizes $(h \circ r)^{-1}(\overline{m}) = r^{-1}(m + n \cdot \N)$. Now assume that $m \geq n$. Since $m' < n$, by the above, we can construct effectively a finite-state $\Sigma$-tree automaton which recognizes $r^{-1}(m' + n \cdot \N)$. Moreover, by Lemma~\ref{cor:inverse_n_recognizable}(2), for each $j \in [0,k-1]$, we can also construct effectively a finite-state $\Sigma$-tree automaton which recognizes $r^{-1}(m' + n \cdot j)$. Lastly, $\Sigma$-tree languages are  closed effectively under union and subtraction, and thus, we can construct effectively a finite-state $\Sigma$-tree automaton which recognizes $r^{-1}(m + n \cdot \N)$.
\end{proof}

\begin{lemma} \rm \label{lm:compute-c(A)}
Let $B$ be computable. If a $(\Sigma,B)$-wta $\A=(Q,\delta,F)$ is given effectively and $\rH(\A)$ is finite, then we can compute the sets $\rH(\A)$ and $\rC(\A)$.
\end{lemma}
\begin{proof} First we prove that the set $\rH(\A)$ can be computed.
For every $i \in \N$ and $q \in Q$ let
\[H_{i, q} = \{\wt(\xi,\rho) \mid \xi \in \T_\Sigma, \hgt(\xi) \leq i, \rho \in \R_\A(q,\xi)\}\enspace.\]
Clearly, we have $H_{0,q} \subseteq H_{1,q} \subseteq \ldots \subseteq \rH(\A)$ for each $q \in Q$.
We claim that, for each $i \in \N$, 
\[\text{if $\forall q\in Q$: $H_{i,q} = H_{i+1,q}$, then $\forall q\in Q$: $H_{i+1,q} = H_{i+2,q}$}.
\]
To show this, let $i \in \N$, $q \in Q$, and $b \in H_{i+2,q}$. There exist  $\xi \in \T_\Sigma$ and $\rho \in \R_\A(q,\xi)$ such that $\hgt(\xi) \leq i+2$ and $\wt(\xi,\rho) = b$. We may assume that $\hgt(\xi) = i+2$. Hence $\xi=\sigma(\xi_1,\ldots,\xi_k)$  such that $\hgt(\xi_j)\leq i+1$ for each $j \in [k]$. Clearly, for each $j \in [k]$, we have $\wt(\xi_j,\rho|_j) \in H_{i+1,\rho(j)}$, so by our assumption there exist   $\zeta_j \in \T_\Sigma$ with $\hgt(\zeta_j)\leq i$ and  run $\theta_j \in \R_\A\big(\rho(j),\zeta_j\big)$
such that $\wt(\xi_j,\rho|_j) = \wt(\zeta_j,\theta_j)$.

Now let $\zeta=\sigma(\zeta_1,\ldots,\zeta_k)$. Obviously, $\hgt(\zeta) \le i+1$. Moreover, let $\theta \in \R_\A(q,\zeta)$ such that $\theta|_j = \theta_j$ for each $j \in [k]$. Clearly, $\wt(\zeta,\theta) \in H_{i+1,q}$, and we calculate
\begin{align*}
    \wt(\zeta,\theta)
    &= \Big( \bigotimes_{j = 1}^k \wt(\zeta_j,\theta|_j) \Big) \otimes \delta_k\big(\theta(1) \cdots \theta(j),\sigma,q\big)\\
    &= \Big( \bigotimes_{j = 1}^k \wt(\xi_j, \rho|_j) \Big) \otimes \delta_k\big(\rho(1) \cdots \rho(j),\sigma,q\big) = \wt(\xi,\rho) = b\enspace. 
\end{align*}
This shows that $b \in H_{i+1,q}$, proving our claim.

We recall that $H_{0,q} \subseteq H_{1,q} \subseteq \ldots \subseteq \rH(\A)$ for each $q \in Q$.
Since $B$ is computable,  we can compute $H_{i,q}$ for every $i\in \N$ and $q\in Q$. Then, 
since $\rH(\A)$ is finite, by computing $H_{0,q}$ for each $q\in Q$, $H_{1,q}$ for each $q\in Q$, and so on, we can find the least number $i_m\in \N$ such that
$H_{i_m,q}=H_{i_m+1,q}$ for each $q\in Q$ and thus by the implication shown above $H_{i_m,q}=H_{j,q}$ for every $q\in Q$ and $j\in \N$ with $j \geq i_m$.
We show that $\rH(\A) = \bigcup_{q \in Q} H_{i_m, q}$. For this, let $b\in \rH(\A)$,
i.e., $b=\wt(\xi,\rho)$ for some $\xi \in \T_\Sigma$ with $\hgt(\xi)=j$, $q\in Q$ and $\rho \in \R_\A(q,\xi)$. Then $b\in H_{j,q}=H_{i_m,q}$. The other inclusion is obvious. Since we can compute the set $\bigcup_{q \in Q} H_{i_m, q}$, the set $\rH(\A)$ can be computed.

Now we prove that the set $\rC(\A)$ can be computed. Let $i_m$ be the number as before.  It suffices to show that
\[\rC(\A)=\{\wt(\xi,\rho)\otimes F_{\rho(\varepsilon)}\mid \xi \in \T_\Sigma, \hgt(\xi) \leq i_m, \rho \in \R_\A(\xi)\}\enspace,\]
because we can compute the set on the right-hand side of the above equality. Let us denote this set by $C$. It is obvious that  $C\subseteq \rC(\A)$. For the proof of the other inclusion,   let $b\in \rC(\A)$,
i.e., $b=\wt(\xi,\rho)\otimes F_q$ for some $\xi \in \T_\Sigma$, $q\in Q$, and  $\rho \in \R_\A(q,\xi)$. Since $\wt(\xi,\rho)\in \rH(\A)$, by the proof of computing the set $\rH(\A)$, we have $\wt(\xi,\rho)\in H_{i_m, q}$, i.e., there exist $\xi' \in \T_\Sigma$ with  $\hgt(\xi') \leq i_m$, and  $\rho' \in \R_\A(q,\xi')$ such that $\wt(\xi,\rho)=\wt(\xi',\rho')$. Hence $b\in C$.
\end{proof}

Next we present the main result of this section. It gives a structural condition on a $(\Sigma,B)$-wta $\A$, for arbitrary strong bimonoid $B$, which is sufficient to imply that $\A$ has the finite-image property. Our result generalizes \cite[Thm.~6.2(a)]{drogoemaemei11} and \cite[Thm. 11]{drostuvog10} 
from bi-locally finite strong bimonoids to arbitrary strong bimonoids, in case of
\cite[Thm. 11]{drostuvog10} even from strings to trees.

 \begin{theorem} \label{thm:sufficient-conds-ensure-rec-step-map}
   Let $\A=(Q,\delta,F)$ be a $(\Sigma,B)$-wta such that $\rH(\A)$ is finite. If, for each $b \in \rC(\A)$, the mapping $f_{\A,b}$ is bounded or $b$ has finite additive order, then the following statements hold.
   \begin{compactenum}
        \item $\A$ has the finite-image property and the preimage property.
        \item If $B$ is computable and $\A$ is given effectively, then we can construct effectively  a crisp-deterministic $(\Sigma,B)$-wta $\B$ such that $\sem{\B} = \sem{\A}$. 
   \end{compactenum} 
 \end{theorem}

 \begin{proof}
 \underline{Proof of 1:} We note that $\rC(\A)$ is finite because $\rH(\A)$ is finite. For each $b \in \rC(\A)$ we define the $(\Sigma,\N)$-wta  $\A'_b=(Q',\delta',F'_b)$ as follows: $Q' = Q \times \rH(\A)$ and for every $k\in \N$, $\sigma \in \Sigma^{(k)}$ and $(q_1,y_1),\ldots,(q_k,y_k), (q,y) \in Q'$, let 
  \[\delta'_k\big((q_1,y_1)\cdots(q_k,y_k),\sigma,(q,y)\big) = 
  \begin{cases}
      1 &\text{ if $\big(\bigotimes_{i=1}^k y_i\big) \otimes \delta_k(q_1 \cdots q_k, \sigma, q) = y$}\\
      0 &\text{ otherwise,} 
  \end{cases}\]
  and let
  \[(F'_b)_{(q,y)} = 
  \begin{cases}
      1 &\text{ if $y \otimes F_q = b$}\\
      0 &\text{ otherwise.} 
  \end{cases}
  \]
  Let $\xi \in \T_\Sigma$ and $b \in \rC(\A)$. 
 We observe that there exists a bijection between the two sets
  \[\{ \rho \in \R_\A(\xi) \mid \wt_\A(\xi,\rho)\otimes F_{\rho(\varepsilon)} = b \} \text{ and } \{ \rho' \in \R_{\A'_b}(\xi)\big)\mid \wt_{\A'_b}(\xi,\rho') \cdot (F'_b)_{\rho'(\varepsilon)} = 1\}.\]
  It follows that $\sem{\A'_b}(\xi) =  f_{\A,b}(\xi)$, and thus, by Equality~\eqref{eq:semantics=sum-f}, we have
  \begin{align*}
\sem{\A}(\xi) =  \bigoplus_{b \in \rC(\A)} \big(\sem{\A'_b}(\xi)\big)b.
\end{align*}
Let us define the mapping $r_b: \T_\Sigma \to B$ by $r_b(\xi) = \big(\sem{\A'_b}(\xi)\big)b$ for each $\xi\in \T_\Sigma$. Then $\sem{\A} = \bigoplus_{b \in \rC(\A)}r_b$ and it suffices to show that $r_b$ is a recognizable step mapping for each $b \in \rC(\A)$ because, obviously, recognizable step mappings are closed under the operation $\oplus$. 

To prove this latter, let $b \in \rC(\A)$. We distinguish the following two cases.

\underline{Case 1:} The mapping $f_{\A,b}$ is bounded, i.e.,  there exists $K \in \N$ such that $f_{\A,b}(\xi) \leq K$ for each $\xi \in \T_\Sigma$. Clearly, $\im(r_b) \subseteq \{jb \mid j \in [0,K]\}$. For each $j \in [0,K]$, let $L_{b,j} = \sem{\A'_b}^{-1}(j)$.
By Lemma~\ref{cor:inverse_n_recognizable}(1), $L_{b,j}$ is a  recognizable $\Sigma$-tree language. By our assumption,  we have $\bigcup_{j \in [0,K]} L_{b,j} = \T_\Sigma$. Hence $r_b = \bigoplus_{j \in [0,K]} (jb) \otimes \1_{(B,L_{b,j})}$, i.e., it is a recognizable step mapping. 

\underline{Case 2:} $b$ has finite additive order. Then we have $\langle b \rangle_\oplus = \{0, b, 2b, \ldots, (i(b)+p(b)-1)b\}$. So
  \[(\forall n \in \N)(\exists \text{ exactly one } j \in [0,i(b)+p(b)-1]) : nb = jb.\]
  Now let $L_{b,j} = \{\xi \in \T_\Sigma \mid \big(\sem{\A'_b}(\xi)\big)b = jb\ \}$ for each $j \in [0,i(b)+p(b)-1]$. Observe that $L_{\0,0}=\T_\Sigma$. We claim that $L_{b,j}$ is recognizable. We have
  \begin{itemize}
      \item  $L_{b,j} =\sem{\A'_b}^{-1}(j)$ if $0 \leq j < i(b)$, and
      \item  $L_{b,j} =\sem{\A'_b}^{-1}(j + p(b)\cdot\N)$ if $i(b) \leq j \leq i(b)+p(b)-1$.
  \end{itemize} 
Then $L_{b,j}$ is recognizable in both cases, by Lemmas~\ref{cor:inverse_n_recognizable}(1) and \ref{lm:inverse_m+nN_recognizable}(1), respectively.

  Let $\xi \in \T_\Sigma$. By the above, there exists a unique number $j \in [0,i(b)+p(b)-1]$ such that $r_b(\xi) = \big(\sem{\A'_b}(\xi)\big)b = jb$, and so $\xi \in L_{b,j}$.
  
  Hence,
  \[r_b = \bigoplus_{0 \leq j \leq i(b)+p(b)-1} (jb) \otimes \mathbb{1}_{(B,L_{b,j})}.\]
  i.e., it is a recognizable step mapping.

\underline{Proof of 2:} By Lemma~\ref{lm:compute-c(A)}, we can compute the set $\rH(\A)$, and thus the set $\rC(\A)$. Moreover, for each $b \in \rC(\A)$, we can construct effectively the $(\Sigma,\N)$-wta~$\A'_b$.

Next we decide, for each $b \in \rC(\A)$, whether the mapping $f_{\A,b}$ is bounded or $b$ has finite additive order.
Note that one of these conditions holds by our assumption.

For this, we run the following two algorithms in parallel for $i =0,1,2,\ldots$. In the first algorithm, we construct effectively the finite $\Sigma$-tree automaton which recognizes the $\Sigma$-tree language $L_{b,i}=\sem{\A'_b}^{-1}(i)$ (cf. Lemma \ref{cor:inverse_n_recognizable}(2)) and check whether $\bigcup_{j \in [0,i]} L_{b,j} = \T_\Sigma$ (cf. \cite[Thm.~2.10.3]{gecste84}).
If this is the case, then it means that $f_{\A,b}$ is bounded by $i$. We let $K_b=i$ and stop. 

In the second algorithm, we compute the sum $i b$ and check whether $i b=jb$ for some $j < i$. If this is the case, then $\langle b \rangle_\oplus = \{0,b,2b,\ldots,(i-1)b\}$ is a finite set. We let $K_b=i-1$ and stop. 

By our note above, the decision algorithm will stop for some $i \in\N$.

If the first algorithm stops, then  we can describe $r_b$ as in Case 1 of Statement 1.  If the second algorithm stops, then we
can compute $i(b)$ and $p(b)$ and can describe $r_b$ as in Case 2 of Statement 1. By Lemmas~\ref{cor:inverse_n_recognizable}(2) and \ref{lm:inverse_m+nN_recognizable}(2), we can construct effectively a finite $\Sigma$-tree automaton which recognizes $L_{b,j}$. If both algorithms stop, then we can proceed in either way. 

Now we have
\[\sem{\A} = \bigoplus_{b \in \rC(\A)} \bigoplus_{0 \leq j \leq K_b} (jb) \otimes \1_{(B,L_{b,j})}\enspace.\]
By applying the direct product construction in the proof of (iv) $\Rightarrow$ (i) of \cite[Lm.~5.3]{fulkosvog19} we can construct effectively the crisp-deterministic $(\Sigma,B)$-wta $\B$ such that $\sem{\B} = \sem{\A}$. 
\end{proof}

Now we can give a simple condition on the strong bimonoid  $B$  and
a structural condition on the wta  $\A$ ensuring  that  $\A$  has the finite image property and the preimage property.

\begin{corollary}\rm \label{cor:small-loops->cA-finite}
    Let $B$  be commutative or one-product free and let small loops of  $\A$ have weight~$\1$. 
 \begin{compactenum}
 \item     Then  $\rH(\A)$ is finite. 
 \item If, in addition, for each $b \in \rC(\A)$, the mapping $f_{\A,b}$ is bounded or $b$ has finite additive order, then $\A$ has the finite-image property and the preimage property.
  \end{compactenum}
\end{corollary}

\begin{proof} First we prove Statement 1. If small loops of $\A$ have weight  $\1$, then by Lemma \ref{lm:weight-of-reduced-tree} we have
    \[\rH(\A) = \{\wt(\xi,\rho) \mid \xi \in \T_\Sigma, \hgt(\xi) < |Q| \text{ and } \rho \in \R_\A(\xi)\}\enspace.\]
     Hence $\rH(\A)$ is finite.  Then Statement 2 follows from Theorem \ref{thm:sufficient-conds-ensure-rec-step-map}(1).
\end{proof}

Next, we compare Theorem \ref{thm:sufficient-conds-ensure-rec-step-map} and \cite[Thm.~7.3]{fulkosvog19}. This makes sense because both results show sufficient conditions for a wta $\A$ such that $\sem{\A}$ is a recognizable step mapping. Let $\A=(Q,\delta,F)$ be a $(\Sigma,B)$-wta. Theorem 7.3 of \cite{fulkosvog19} requires that $\A$ has finite order property, i.e., (a) the set $\langle \im(\delta)\rangle_{\otimes}$ is finite and (b) each element $b \in \langle \im(\delta)\rangle_{\otimes} \otimes \im(F)$ has finite additive order. Condition (a) implies that $\rH(\A)$ is finite. And Condition (b) implies that each element $b \in \rC(\A)$ has finite additive order. Hence Theorem \ref{thm:sufficient-conds-ensure-rec-step-map} is at least as strong as \cite[Thm.~7.3]{fulkosvog19}. The next example shows a scenario in which Theorem \ref{thm:sufficient-conds-ensure-rec-step-map} is applicable but not \cite[Thm.~7.3]{fulkosvog19}.

\begin{example}  \label{ex:reachable-values-versus-closure} \rm 
We consider the ranked alphabet $\Sigma = \{\gamma^{(1)}, \nu^{(1)}, \alpha^{(0)}\}$ and the arctic semiring $\mathrm{ASR}_\N=(\N_{-\infty},\max,+,-\infty,0)$. Moreover, we let $\A = (Q,\delta,F)$ be the trim $(\Sigma,\mathrm{ASR}_\N)$-wta where $Q =\{q_1,q_2\}$, $\delta_0(\varepsilon,\alpha,q_1)=\delta_1(q_1,\gamma,q_1)=0$, 
and $\delta_1(q_1,\nu,q_2)=1$; and $F(q_1)=F(q_2)=0$. 

Since  $\langle \im(\delta)\rangle_{+} = \langle \{0,1\} \rangle_{+} = \mathbb{N}$ is infinite, we cannot apply \cite[Thm.~7.3]{fulkosvog19}. Moreover, since $\mathrm{ASR}_\N$ is one-product free, small loops of $\A$ have weight $0$ and each $n \in \rC(\A)$ has finite additive order, by  Corollary \ref{cor:small-loops->cA-finite}(2) we obtain that $\sem{\A}$ is a recognizable step mapping.
  \qed
  \end{example}

As a consequence of Theorem \ref{thm:sufficient-conds-ensure-rec-step-map}, we can extend Lemma \ref{lm:inverse_b_recognizable_with_B_finite_semiring} from finite semirings to finite strong bimonoids.

\begin{corollary} \rm \label{cor:inverse_b_recognizable_with_B_finite}
     Let $B$ be finite and $\A$ be a $(\Sigma,B)$-wta. Then the following statements hold.
    \begin{compactenum}
        \item $\A$ has the preimage property.
        \item If $B$ is computable and $\A$ is given effectively, then, for each $b \in B$, we can construct effectively a finite-state $\Sigma$-tree automaton which recognizes $\sem{\A}^{-1}(b)$.
    \end{compactenum}
\end{corollary}

\begin{proof} Let $\A=(Q,\delta,F)$ and we abbreviate $\sem{\A}$ by $r$. Since $B$ is finite, so is the set $\rH(\A)$. Moreover, $b$ has finite additive order for each $b \in \rC(\A)$. Then Statement 1 follows from Theorem \ref{thm:sufficient-conds-ensure-rec-step-map}(1). 

Now we prove Statement 2. By Theorem \ref{thm:sufficient-conds-ensure-rec-step-map}(2), we can construct effectively a crisp-deterministic $(\Sigma,B)$-wta $\B=(Q',\delta',F')$ such that $\sem{\B}=\sem{\A}$. 
Note that $\im(\sem{\A})\subseteq \im (F')$ because $\B$ is crisp-deterministic. Then Statement~2 follows from the remark after Lemma~\ref{lm:cdwta_finite_image}.
\end{proof}

Finally, as an application of Corollary \ref{cor:inverse_b_recognizable_with_B_finite}, we show that for arbitrary past-finite monotonic strong bimonoid $B$, every $(\Sigma,B)$-wta has the finite-image property. This generalizes the preimage property \cite[III. Cor.~2.5]{berreu88} from strings to trees and from the semiring $\N$ to  past-finite monotonic strong bimonoids.

\begin{theorem}  \label{thm:inverse_b_recognizable}
    Let $B$ be past-finite monotonic and $\A=(Q,\delta,F)$ be a $(\Sigma,B)$-wta. Then the following statements hold.
    \begin{compactenum}
        \item $\A$ has the preimage property.
        \item  If $B$ is computable, 
$\A$ is given effectively, and $b \in B$ such that the set $\past(b)$ is computable, then we can construct effectively  a finite-state $\Sigma$-tree automaton which recognizes $\sem{\A}^{-1}(b)$.
    \end{compactenum}
\end{theorem}

\begin{proof}
    \underline{Proof of 1:} Let $b \in B$ and put  $C = B \setminus \past(b) = \{a \in B\mid a \npreceq b\}$. Moreover, let $\sim$ be the equivalence relation on the set $B$ defined such that its classes are the singleton sets $\{a\}$ for each $a \in \past(b)$ and the set $C$. We claim that $\sim$  is a congruence. To show that $C$ is a congruence class, let  $c,c' \in C$  and  $d \in B$. Since  $B$  is monotonic, we have  $c \preceq c \oplus d$, hence  $c \oplus d \in C$  and similarly  $c' \oplus d \in C$.
     Also, if $d \ne \0$, again we obtain $c \preceq c \otimes d$  and  $c \preceq d\otimes c$, showing $c\otimes d,d\otimes c \in C$  and similarly  $c'\otimes d,d\otimes c' \in C$. Hence $C$ is a congruence class and the relation $\sim$ is a congruence on the strong bimonoid $B$. Then the quotient strong bimonoid $B/_\sim$ is finite. Let $h: B \to B/_\sim$ be the canonical strong bimonoid homomorphism. Let us abbreviate $\sem{\A}$ by $r$. Then, by Lemma \ref{lm:closure-homomorphism}, $(h \circ r) \in \mathrm{Rec}(\Sigma,B/_\sim)$. Moreover $r^{-1}(b) = (h \circ r)^{-1}(\{b\})$.  Since $B/_\sim$ is finite, by Corollary \ref{cor:inverse_b_recognizable_with_B_finite}(1), the $\Sigma$-tree language $(h\circ r)^{-1}(\{b\})$ is recognizable.

    \underline{Proof of 2:} Let $\sim$ be the congruence  defined as in the proof of Statement 1 and $h: B \to B/_\sim$ be the canonical strong bimonoid homomorphism. Since $B$ is computable and also $\past(b)$ is computable, we can give effectively  the congruence classes of $\sim$, i.e., the elements of $B/_\sim$, by choosing only one representative for each congruence class. By Lemma~\ref{lm:closure-homomorphism}, we can construct effectively the $(\Sigma,B/_\sim)$-wta $h(\A)$ such that $\sem{h(\A)}= h \circ r$. Since $B/_\sim$ is computable and finite, and $\A$ is given effectively, by Corollary \ref{cor:inverse_b_recognizable_with_B_finite}(2), we can construct effectively  a finite-state $\Sigma$-tree automaton which recognizes $(h\circ r)^{-1}(\{b\})$.
\end{proof}

\begin{corollary} \rm Let $B$ be past-finite monotonic and $\A=(Q,\delta,F)$ be a $(\Sigma,B)$-wta. Then $\supp(\sem{\A})$ is a recognizable $\Sigma$-tree language.
\end{corollary}
\begin{proof} By Theorem \ref{thm:inverse_b_recognizable}, the $\Sigma$-tree language $\sem{\A}^{-1}(\0)$ is recognizable. Since the class of recognizable $\Sigma$-tree languages is closed under complement and $\supp(\sem{\A}) = \T_\Sigma \setminus  \sem{\A}^{-1}(\0)$, we obtain the statement of the corollary.
  \end{proof}

\section{Characterization and decidability of the finite-image property}

In most of this section,  $B$  will be a past-finite strong bimonoid. Our main results of this section will describe when a $(\Sigma,B)$-wta $\A$ has the finite-image property and when it can decided whether $\A$ has the finite-image property. First we characterize when an arbitrary $(\Sigma,B)$-wta $\A$  has the finite-image property by structural properties of the wta $\A$. 

\begin{theorem}  \label{cor:characterization-of-crisp-determinizability} Let $B$ be a past-finite monotonic strong bimonoid and 
$\A$ be a trim $(\Sigma,B)$-wta. Then the following statements are equivalent.
\begin{compactenum}
\item $\A$  has the finite-image property.
\item Small loops of $\A$ have weight $\1$ and, for each $b \in \rC(\A)$, the mapping  $f_{\A,b}$  is bounded or $b$ has finite additive order.
\end{compactenum} 
\end{theorem}
\begin{proof} \underline{(1) $\Rightarrow$ (2):} First we show that small loops of $\A$  have weight $\1$. We proceed by contraposition. Suppose there exist  $q \in Q$, $c \in \C_\Sigma$, and $\rho \in \R_\A(q,c,q)$ such that $\hgt(c) < |Q|$ and  $\1 \prec \wt(c,\rho)$. Since $\A$ is trim, the state $q$ is useful and thus there exist  $\xi \in \T_\Sigma, \theta \in \R_\A(q,\xi)$ and $c' \in \C_\Sigma$, $q' \in Q$ with $F_{q'}\ne\0$, and $\rho' \in \R_\A(q',c',q)$. By Lemma~\ref{lm:decomposition-of-a-run}, for each $n \in \N$, we have
\[
 \wt(c'\big[c^{n}[\xi]\big], \rho'\big[\rho^{n}[\theta]\big]) =
  l_{c',\rho'} \otimes (l_{c,\rho})^n \otimes \wt(\xi,\theta) \otimes (r_{c,\rho})^n \otimes r_{c',\rho'}\enspace.
\]
Since $\1 \prec \wt(c,\rho) = l_{c,\rho}\otimes r_{c,\rho}$ , we have $\1 \prec l_{c,\rho}$  or $\1 \prec r_{c,\rho}$  and thus by monotonicity we obtain
\begin{equation}\label{eq:prec-sequence}
  \wt  (c'\big[c^0[\xi]\big], \rho'\big[\rho^0[\theta]\big]  ) \prec \wt  (c'\big[c^1[\xi]\big], \rho'\big[\rho^1[\theta]\big]  ) \prec \ldots \enspace.
\end{equation}

We define a sequence $\xi_1, \xi_2, \xi_3, \ldots$ 
of trees in $\T_\Sigma$ such that the elements $\sem{\A}(\xi_1)$, $\sem{\A}(\xi_2)$, $\sem{\A}(\xi_3)$, $\ldots$ are pairwise different as follows.
We let $\xi_1 = c'\big[c[\xi]\big]$. Then $P_1 = \past(\sem{\A}(\xi_1))$ is finite. By \eqref{eq:prec-sequence} we choose $n_2$ such that $\wt(c'\big[c^{n_2}[\xi]\big], \rho'\big[\rho^{n_2}[\theta]\big]) \not\in P_1$ and let $\xi_2=c'\big[c^{n_2}[\xi]\big]$. Since $\rho'\big[\rho^{n_2}[\theta]\big] \in \R_\A(q',\xi_2)$ and $B$ is monotonic, we have
\[\wt(\xi_2, \rho'\big[\rho^{n_2}[\theta]\big]) \preceq \wt(\xi_2, \rho'\big[\rho^{n_2}[\theta]\big])\otimes F_{q'} \preceq \sem{\A}(\xi_2).
\]
(Note that $F_{q'}$ may be $\1$.) Hence $\sem{\A}(\xi_2) \notin P_1$. Put $P_2 = \past(\sem{\A}(\xi_2))$. 
Then we choose $n_3 \in \N$ such that  $\wt(c'\big[c^{n_3}[\xi]\big], \rho'\big[\rho^{n_3}[\theta]\big]) \not\in P_1 \cup P_2$ and let $\xi_3=c'\big[c^{n_3}[\xi]\big]$. As before,  we have
\(\sem{\A}(\xi_3) \notin P_1 \cup P_2\).
Continuing this process, we obtain the desired sequence of trees. It means that $\A$ does not have the finite-image property.

Now let $b\in \rC(\A)$. If the mapping $f_{\A,b}$ is not bounded, then there exists an infinite sequence $\xi_1,\xi_2,\ldots$ of trees in $\T_\Sigma$ such that $f_{\A,b}(\xi_1) < f_{\A,b}(\xi_2) < \ldots$. By Equality \eqref{eq:semantics=sum-f}, we have $\big(f_{\A,b}(\xi_i)\big)b \preceq \sem{\A}(\xi_i)$ for each $i\in \N$. Thus $\big(f_{\A,b}(\xi_i)\big)b \in P$, where
$P=\bigcup_{a\in \im(\sem{\A})}\past(a)$. Since  $\im(\sem{\A})$ is finite and $B$ is past-finite, the set $P$ is also finite. 
Hence $\big(f_{\A,b}(\xi_i)\big)b=\big(f_{\A,b}(\xi_j)\big)b$ for some $i,j\in \N$ with $i<j$, which implies that $b$ has finite additive order. 

\underline{(2) $\Rightarrow$ (1):} It follows from Corollary \ref{cor:small-loops->cA-finite}(2). 
\end{proof}

\begin{figure}[t]
\centering
\begin{tikzpicture}
\node at (0,0) (b1) {$\0$};
\node at (0,0.65) (b2) {$\1$};
\node[anchor=south] at (0,1) (b3) {$1$};
\node[anchor=south] at (0,1.7) (b4) {$2$};
\node[anchor=south] at (0,2.4) (b5) {$\vdots$};
\node[anchor=south] at (0,3.2) (b6) {$1'$};
\node[anchor=south] at (0,3.9) (b7) {$2'$};
\node[anchor=south] at (0,4.6) (b8) {$\vdots$};

\draw[-]
  (b1) -- (b2)
  (b2) -- (b3)
  (b3) -- (b4)
  (b4) -- (b5)
  (0,3) -- (b6)
  (b6) -- (b7)
  (b7) -- (b8)
;
\draw[decorate,decoration={brace,amplitude=2.5pt,mirror},xshift=10pt]
  (0.15,1.1) -- (0.15,2.9) node[midway,xshift=12.5pt] {$\N_+$};
\draw[decorate,decoration={brace,amplitude=2.5pt,mirror},xshift=10pt]
  (0.15,3.2) -- (0.15,5.1) node[midway,xshift=12.5pt] {$\N'_+$};
\end{tikzpicture}
\caption{\label{fig:ordering-B} The Hasse diagram of the ordering $(B,\leq_B)$ in Example \ref{ex:cA-infin-imA-fin}.}
\end{figure}

The following example shows that in Theorem \ref{cor:characterization-of-crisp-determinizability} even for commutative semirings we cannot replace the assumption that $B$  is past-finite by being idempotent. The wta $\A$ given below can actually be considered as a weighted string automaton (cf. Section~\ref{sect:string-case}).

\begin{example} \label{ex:cA-infin-imA-fin} \rm
In this example we give an idempotent and monotonic semiring $B$ which is not past-finite, and a $(\Sigma,B)$-wta $\A$ such that $\A$ has the finite-image property.

For this, let $B=\N_+ \cup \N'_+ \cup \{\0,\1\}$ where $\N'_+$ is an isomorphic copy of $\N_+$, and furthermore, $\0$ and $\1$ are new elements such that $\{\0,\1\} \cap (\N_+ \cup \N'_+) = \emptyset$. For each $C \in \{\N_+,\N'_+\}$ we denote by $\leq_C$ the usual linear order of $C$ and $+_C$ the usual binary addition operation on $C$. 

We define a partial ordering $\leq_B$ on $B$ such that $\leq_B$ extends each of the linear orders of $\N_+$ and $\N'_+$ (i.e., $\leq_B \cap (C \times C) = \leq_C$ for each $C \in \{\N_+,\N'_+\}$) and such that $\0 <_B \1 <_B a <_B b'$ for every $a \in \N_+$ and $b' \in \N'_+$. Note that $(B,\leq_B)$ is a partial order. Figure \ref{fig:ordering-B} shows the Hasse diagram of the ordering $(B,\leq_B)$.
Moreover, letting $\vee$ be the usual binary supremum operation, $(B,\vee)$ is a join-semilattice. Note that we have $a \vee b' = b'$ for each $a \in \N_+$ and $b' \in \N'_+$.

Now we define a commutative multiplication $\otimes$ on $B$ as follows. For each $C \in \{\N_+,\N'_+\}$ we let $\otimes|_{C \times C} = +_C$. Furthermore, for every $a \in \N_+$ and $b' \in \N'_+$ we let $a \otimes b' = a' +_{\N'_+} b'$. Finally, let $\0 \otimes b = \0$ and $\1 \otimes b = b$ for each $b \in B$.
Then $(B,\vee,\otimes,\0,\1)$ is a monotonic strong bimonoid. Clearly, $B$ is not past-finite, e.g., $\past(1') = (\{\0,\1\} \cup \N_+)$, which is infinite. However, $B$ is idempotent because it is a join-semilattice. Obviously, $B$ is monotonic.

Let $\Sigma=\{\gamma^{(1)},e^{(0)}\}$. Next we consider the $(\Sigma,B)$-wta $\A=(\{p,q\},\delta,F)$ with $\delta_0(\varepsilon,e,p) = \1$, $\delta_0(\varepsilon,e,q)=1'$ (in $\N'_+$), $\delta_1(p,\gamma,p)=1$ (in $\N_+$), $\delta_1(q,\gamma,q) = \1$, $\delta_1(p,\gamma,q)=\0=\delta_1(q,\gamma,p)$,  and $F_p=\1=F_q$. Let $n \in \N$ and $\xi=\gamma^n(e) \in \T_\Sigma$. Clearly, there exist  two runs on $\xi$: Let us denote them by $\rho^{\xi,p}$ and $\rho^{\xi,q}$. Obviously, $\wt(\xi,\rho^{\xi,p})=n$ if $n \geq 1$, and otherwise $\wt(\xi,\rho^{\xi,p})=\1$, and furthermore, $\wt(\xi,\rho^{\xi,q})=1'$. Thus, we have $\sem{\A}(\xi) = \wt(\xi,\rho^{\xi,p}) \vee \wt(\xi,\rho^{\xi,q}) = 1'$. Hence, $\A$ has the finite-image property.

Observe that  $\A$  has two small loops with weight $1$ (in $\N_+$) respectively $\1$.
 \qed
\end{example}

We say that the strong bimonoid  $B$ \emph{has effective tests for $\0$ and $\1$}  if for each  $b \in B$  we can decide whether  $b = \0$  and whether  $b = \1$.

\begin{theorem} \label{cor:image-A-finite-decidable} \cite[Cor.~14]{drofulkosvog20}
  Let $B$ be past-finite monotonic and have effective tests for $\0$ and $\1$. Moreover, let $\A$ be given effectively.   If $B$ is additively locally finite or $\A$ is finitely ambiguous, then it is decidable whether $\A$ has the finite-image property.
\end{theorem}

\begin{proof} 
By Theorem~\ref{thm:A'-equivalent-to-A}, we may assume that $\A$ is trim.
  By Theorem~\ref{cor:characterization-of-crisp-determinizability}, $\A$ has the finite-image property if and only if small loops of $\A$ have weight $\1$.
The latter property is decidable
because (a) there exist  only finitely many $c \in \C_\Sigma$  such that  $\hgt(c) < |Q|$,
and (b) since  $B$ is monotonic, for all  $c \in \C_\Sigma$, $q \in Q$, and  $\rho \in \R_\A(q,c,q)$  we have $\wt(c,\rho) = \1$
if and only if for each  $v \in \pos(c)$  we have
$\delta_k(\rho(v1)\cdots\rho(vk),\sigma,\rho(v)) = \1$  where
$\sigma = c(v)$  and  $k = \rk(\sigma)$, and (c) this is decidable because $B$ has an effective test for $\1$.
\end{proof}

The decidability problem addressed in Theorem~\ref{cor:image-A-finite-decidable} is meaningful, because in Example~\ref{ex:commutative-SB-not-crisp-determinizable} we considered the additively locally finite and past-finite monotonic semiring $\mathrm{ASR}_\N$ and a deterministic $(\Sigma,\mathrm{ASR}_\N)$-wta $\A$ for which 
$\im(\sem{\A})$ is infinite. 

As an immediate consequence of Theorems~\ref{thm:inverse_b_recognizable}(1)~and~\ref{cor:image-A-finite-decidable} and Lemma~\ref{lm:cdwta_finite_image}, under the assumptions of Theorem~\ref{cor:image-A-finite-decidable}, it is decidable whether $\A$  is crisp-determinizable.

Next, we compare Theorem~\ref{thm:sufficient-conds-ensure-image-finite} with 
Theorem~\ref{cor:image-A-finite-decidable} in the following sense: we show an example of a wta $\A$ such that (a) by applying Theorem~\ref{thm:sufficient-conds-ensure-image-finite} we know that $\A$ has the finite-image property and (b) we cannot apply Theorem~\ref{cor:image-A-finite-decidable} to decide whether $\A$ has the finite-image property.

\begin{example} \label{ex:suffient-but-not-decidable} \rm
We consider the ranked alphabet $\Sigma = \{\gamma^{(1)}, \nu^{(1)}, \alpha^{(0)}\}$ and the tropical semiring $\mathrm{TSR}_\N=(\N_{\infty},\min,+,\infty,0)$. Moreover, we let $\A = (Q,\delta,F)$ be the trim $(\Sigma,\mathrm{TSR}_\N)$-wta (as in Example~\ref{ex:reachable-values-versus-closure}) where $Q =\{q_1,q_2\}$, $\delta_0(\varepsilon,\alpha,q_1)=\delta_1(q_1,\gamma,q_1)=0$, and $\delta_1(q_1,\nu,q_2)=1$; and $F_{q_1}=F_{q_2}=0$.

Then $\A$ satisfies the assumptions of Theorem~\ref{thm:sufficient-conds-ensure-image-finite}. In fact, $\im(\sem{\A}) = \{0,1\}$.
  
  The tropical semiring $\mathrm{TSR}_\N$ cannot be extended into a past-finite monotonic semiring $(\N_{\infty},\min,+,\infty,0,\preceq)$. To see this, assume that $\preceq$ is a monotonic partial order. Then  $a \preceq \min\{a,b\}$  for every $a,b \in \mathbb{N}_{\infty}$ (by Condition (i) of monotonicity). Thus  $a \geq b$ (i.e., $b = \min\{a,b\}$) implies  $a \preceq b$.   However, $\geq$  is not past-finite. Hence, we cannot use Theorem~\ref{cor:image-A-finite-decidable} to decide whether $\A$ has the finite-image property. \qed
  \end{example}

  \begin{theorem}
    Let $B$ be monotonic and have effective tests for $\0$ and $\1$. It is decidable, for arbitrary unambiguous $(\Sigma,B)$-wta $\A$, whether $\A$ has the finite-image property.  
\end{theorem}

\begin{proof}
    By Theorem \ref{thm:A'-equivalent-to-A}, we may assume that $\A$ is trim. Then, by the proof of Theorem~\ref{cor:image-A-finite-decidable}, we can decide whether small loops of $\A$ have weight $\1$. 
    
    If this is not the case, then we follow the proof (1 $\Rightarrow$ 2) of Theorem~\ref{cor:characterization-of-crisp-determinizability} up to \eqref{eq:prec-sequence}, and as there, we can produce an infinite set of weights of runs on trees. Due to unambiguity, these are also the weights of the corresponding trees. Thus, $\im(\sem{\A})$ is infinite, i.e., $\A$ does not have the finite-image property.
    
    Otherwise, by Corollary~\ref{cor:small-loops->cA-finite}, $\A$ has the finite-image property.
\end{proof}

As a side-result, we show that the well-known decidability of finiteness of context-free languages \cite[Thm. 8.2.2]{har78} can be formally derived from Theorem~\ref{cor:image-A-finite-decidable}.

\begin{corollary} \label{cor:finiteness-of-cf-languages}\rm
  The finiteness of context-free languages is decidable. 
\end{corollary}

\begin{proof} Let $G = (N,\Delta,P,S)$ be a context-free grammar. By \cite[Thm.~3.2.3]{har78} we may assume that $G$ is reduced. 
  Then we have
  \begin{equation}
        \text{for each set $U \subseteq \Delta^*$:} \ U \text{ is finite if and only if  } \{|w| \mid w \in U\} \text{ is finite.} \label{eq:fin-fin}
          \end{equation}

 We construct effectively  the ranked alphabet $(P,\rk)$ such that if $p \in P$ has the form
  \[
    A \to w_0 A_1 w_1 \cdots A_k w_k
  \]
  for some $k \in \mathbb{N}$, $w_0,w_1,\ldots,w_k \in \Delta^*$, and $A,A_1,\ldots,A_k \in N$, then $\rk(p)=k$.
  
  Finally, we consider the arctic semiring $\mathrm{ASR}_\N= (\N_{-\infty},\max,+,-\infty,0)$ and construct effectively  the $(P,\mathrm{ASR}_\N)$-wta $\rA(G)=(N,\delta,F)$ by $F_S=0$ and $F_A=-\infty$  for each $A \in N\setminus\{S\}$ and if $p \in P$ has the form as above, then
  \[
\delta_k(A_1\cdots A_k,p,A)= |w_0w_1\cdots w_k|\enspace.
\]
Using \eqref{eq:fin-fin}, it is easy to see that 
\begin{equation}
  \im(\sem{\rA(G)}) \ \text{ is finite if and only if  } \mathrm{L}(G) \text{ is finite.} \label{eq:Afinite-Gfinite}
  \end{equation}

  Since $\mathrm{ASR}_\N$  is  past-finite monotonic and idempotent, and $\rA(G)$ is unambiguous and trim, we obtain from \eqref{eq:Afinite-Gfinite} and Theorem~\ref{cor:image-A-finite-decidable}  that it is decidable whether $\mathrm{L}(G)$ is finite or not.
  \end{proof}

To prove that we can decide whether an arbitrary $(\Sigma,\bbN)$-wta has the finite-image property (cf. Theorem~\ref{thm:bbN-wta-bounded-dec}), we need the following preparation.

To each $(\Sigma,B)$-wta $\cA=(Q,\delta,F)$ we associate the 
 $(\Sigma,\bbB)$-wta   $\fta(\cA)=(Q,\delta_\bbB,F_\bbB)$ defined as follows:
\begin{compactitem}
\item for every $k\in \bbN$, $w \in Q^k$, $\sigma \in \Sigma^{(k)}$, and $q\in Q$, let
$(\delta_\bbB)_k(w,\sigma,q)=1$ if and only if $\delta_k(w,\sigma,q)\ne \0$, and
\item for each $q\in Q$, let
$(F_\bbB)_q= 1$ if and only if $F_q\ne \0$.
\end{compactitem}

\begin{observation} \label{obs:finitely-ambiguous} \rm
For each $\xi\in  \rT_\Sigma$, we have $|\rR_\cA^{F \neq \bb0}(\xi)|= |\rR_{\fta(\cA)}^{F_\bbB \neq \bb0}(\xi)|$ .
\end{observation}

\begin{lemma}\rm \label{lm:wta-bounded-iff-fta-ambiguous} 
Let $\cA$ be a trim $(\Sigma,\bbN)$-wta such that all small loops of $\cA$ have weight $1$. Then $\cA$ has the finite-image property if and only if 
 $\fta(\cA)$ is finitely ambiguous.
\end{lemma}    

\begin{proof} 
By Corollary~\ref{cor:small-loops->cA-finite}(1), the set $\rH(\cA)$ is finite, and thus $\rC(\cA)$ is also finite. 

Let $C = \max\{b \mid b \in \rC(\cA)\}$, i.e., $C$ is the maximum of all possible complete run weights of $\cA$.

First assume that $\fta(\cA)$ is finitely ambiguous under the uniform bound $K \in \bbN$ and let $\xi\in \rT_\Sigma$.  Then, by Observation~\ref{obs:finitely-ambiguous}, $|\rR_\cA^{F \neq \bb0}(\xi)| \le K$, and thus, we have
\[\sem{\cA}(\xi)=\sum_{\rho\in\rR_\cA^{F \neq \bb0}(\xi)}\wt(\xi)\cdot F_\rho(\varepsilon) \le \sum_{\rho\in\rR_\cA^{F \neq \bb0}(\xi)} C\le K\cdot C.\] Hence, $\cA$ has the finite-image property.

Now we prove the other direction. For this, assume that $\cA$ has the finite-image property. Let $K = \max\{\sem{\cA}(\xi) \mid \xi \in \rT_\Sigma\}$. Moreover, let $\xi \in \rT_\Sigma$. Then \[K \geq \sem{\cA}(\xi) = \sum_{\rho\in\rR_\cA^{F \neq \bb0}(\xi)}\wt(\xi)\cdot F_\rho(\varepsilon) \geq \sum_{\rho\in\rR_\cA^{F \neq \bb0}(\xi)} 1 = |\rR_{\cA}^{F \neq 0}(\xi)|\enspace,\]
and Observation~\ref{obs:finitely-ambiguous} shows that $\fta(\cA)$ is finitely ambiguous.
\end{proof}

\begin{lemma} \label{lm:decide-useful-state} \rm
    Let $B$ have an effective test for $\0$. It is decidable, for arbitrary $(\Sigma,B)$-wta $\A=(Q,\delta,F)$ given effectively, whether $\A$ has a useful state. 
\end{lemma}

\begin{proof}
In fact, $\cA$ has a useful state
if and only if there exists a state  $q \in Q$  with $F_q \ne \0$  for which there exist  a tree  $\xi \in \rT_\Sigma$ and a run  $\rho \in \rR_\cA(q,\xi)$. By standard pumping arguments,
the latter is the case if and only if there exist  a tree  $\xi'$ with  $\hgt(\xi') < |Q|$ and run  $\rho' \in \rR_\cA(q,\xi')$. This property is easily decidable.
\end{proof}

    The following result will be crucial for Theorem~\ref{thm:bbN-wta-bounded-dec}.

    \begin{theorem}\label{thm:seidl-theorem} \cite[Thm~2.5(2)]{sei89} Let  $\cA$ be a $(\Sigma,\bbB)$-wta. We can decide whether $\cA$ is finitely ambiguous.
\end{theorem}

    \begin{theorem} \label{thm:bbN-wta-bounded-dec}
        It is decidable, for arbitrary $(\Sigma,\bbN)$-wta $\cA$, whether $\cA$ has the finite-image property.
    \end{theorem}
    
    \begin{proof} 
    First, by Lemma~\ref{lm:decide-useful-state}, we decide whether $\cA$ has a useful state. If $\cA$ does not have a useful state, then $\sem{\cA}=\widetilde{\0}$ and thus $\sem{\cA}$ has the finite-image property. Otherwise, by Theorem~\ref{thm:A'-equivalent-to-A} we may assume that $\cA$ is trim.
    
By the proof of Theorem~\ref{cor:image-A-finite-decidable} we can decide whether 
        small loops of $\cA$ have weight $1$. If the answer is no, then $\cA$ has a loop $\rho$ on a context $c$ with $1 < \wt(c,\rho)$. 
        Then, by Theorem~\ref{cor:characterization-of-crisp-determinizability},  $\im(\sem{\cA})$ is infinite, i.e., $\cA$ does not have the finite-image property. (Note that $\bbN$ is past-finite monotonic.)
        
        Otherwise, all small loops of $\cA$ have weight $1$. Then by Lemma~\ref{lm:wta-bounded-iff-fta-ambiguous}, $\cA$ has the finite-image property if and only if $\fta(\cA)$ is finitely ambiguous. Finally, by Theorem~\ref{thm:seidl-theorem}, we can decide whether $\fta(\cA)$ is finitely ambiguous. This completes the proof.
    \end{proof}
    
    Next we can characterize those past-finite monotonic computable strong bimonoids $B$
for which for each $(\Sigma,B)$-wta  $\A$  it is decidable whether $\A$ has the finite-image property.

    \begin{theorem} \label{thm:characterization}
        Let $B$ be past-finite monotonic and computable. Then the following statements are equivalent:
        \begin{compactenum}
            \item It is decidable, for each ranked alphabet $\Sigma$, whether
            an arbitrary $(\Sigma,B)$-wta given effectively has the finite-image property.
            \item It is decidable, for each $b \in B$, whether $b$ has finite additive order.
        \end{compactenum}
    \end{theorem}

    \begin{proof}
        (1 $\Rightarrow$ 2).  Let $b \in B$. We construct the $(\{\gamma^{(1)}, \alpha^{(0)}\},B)$-wta $\cA=(\{p,q\},\delta,F)$ such that
\begin{compactitem}     
\item    $\delta_1(p,\gamma,q) = b$, $\delta_0(\varepsilon,\alpha,p) = \delta_1(p,\gamma,p) = \delta_1(q,\gamma,q) = F_q = \bb1$, and 
\item $\delta_0(\varepsilon,\alpha,q) = \delta_1(q,\gamma,p) = F_p = \bb0$. 
\end{compactitem}
Clearly, for each $n \in \bbN$, we have $\sem{\cA}(\gamma^n\alpha) = nb$ and thus we have $\im(\sem{\cA}) = \langle b \rangle_\oplus$.
            
         Therefore $\cA$ has the finite-image property if and only if $b$ has finite additive order. By Statement~1, the former is decidable. Hence, it is decidable whether $b$ has finite additive order.

        (2 $\Rightarrow$ 1). Let $\cA$ be an arbitrary $(\Sigma,B)$-wta given effectively. By Theorem \ref{thm:A'-equivalent-to-A}, we may assume that $\cA$ is trim. By the proof of Theorem~\ref{cor:image-A-finite-decidable}
 we can decide whether 
        small loops of $\cA$ have weight $\1$.  If the answer is no, then by Theorem~\ref{cor:characterization-of-crisp-determinizability},  $\im(\sem{\cA})$ is infinite, i.e., $\cA$ does not have the finite-image property. Otherwise, all small loops of $\cA$ have weight $\1$. Then, by Corollary~\ref{cor:small-loops->cA-finite}(1),
the set $\rH(\cA)$ is finite, which implies that the set $\rC(\cA)$ is also finite. In addition, by Lemma~\ref{lm:compute-c(A)} and the fact that $B$ is computable, we can compute the set $\rC(\cA)$.
        Due to our assumption on small loops of $\cA$ and by Theorem~\ref{cor:characterization-of-crisp-determinizability}, $\cA$ has the finite-image property if and only if, for each $b \in \rC(\cA)$, the mapping $f_{\cA,b}$ is bounded or $b$ has finite additive order. Since  $B$  is computable, by the proof of Theorem~\ref{thm:sufficient-conds-ensure-rec-step-map}, we can construct effectively, for each $b \in \rC(\cA)$, the $(\Sigma,\bbN)$-wta  $\A'_b$  described in the proof of Theorem~\ref{thm:sufficient-conds-ensure-rec-step-map}. By Theorem \ref{thm:bbN-wta-bounded-dec}, it is decidable, for each $b \in \rC(\cA)$, whether $\sem{\cA'_b}$ has finite image. Moreover, by Statement~2, it is decidable, for each $b \in \rC(\cA)$, whether $b$ has finite additive order. Therefore, it is decidable, whether $\cA$ has the finite-image property. 
    \end{proof}
    
    \begin{lemma} \label{lm:monotonic-semiring} \rm
        Let $B$ be a monotonic strong bimonoid such that it is left or right distributive. Then either $B$ is idempotent or else each $b \in B \setminus \{\0\}$ has infinite additive order.
    \end{lemma}
    
    \begin{proof}
        We may assume that $B$ is left distributive.
        Clearly, if $\1 = \1 \oplus \1$, then $B$ is idempotent. Therefore, assume that $\1 \neq \1 \oplus \1$, and thus, $\1 \prec \1 \oplus \1$. Since $B$ is monotonic and left distributive,  for each  $b \in B \setminus \{ \0 \}$,  we have  $b \prec b \otimes (\1 \oplus \1) = b \oplus b$.
Consequently,  $b\oplus b \prec (b\oplus b) \oplus (b \oplus b) = 4b$.
Thus, for each  $n \in \N_+$,  we have  $b \prec b \oplus b \prec ... \prec (2^n)b$.
         Hence, $\langle b \rangle_\oplus$ is infinite. (Observe that the proof is similar if we assume right distributivity instead of left distributivity.)
         %
         %
    \end{proof}
    
    \begin{theorem} \label{thm:past-fin-mon-semiring-fin-im-dec}
        Let $B$ be a past-finite monotonic strong bimonoid and computable such that it is left or right distributive. Then it is decidable, for each $(\Sigma,B)$-wta $\cA$ given effectively, whether $\cA$ has the finite-image property. In particular, this is decidable if $B$ is a past-finite monotonic and computable semiring.
    \end{theorem}
    
    \begin{proof}
        Let $\cA$ be an arbitrary $(\Sigma,B)$-wta given effectively. We first check whether $\1=\1\oplus\1$. If this is the case, then $B$ is idempotent, and thus, by Theorem~\ref{thm:characterization}, it is decidable whether $\cA$ has the finite-image property. Otherwise, by Lemma~\ref{lm:monotonic-semiring}, each $b \in B \setminus \{ \0 \}$ has infinite additive order. Again, by Theorem~\ref{thm:characterization}, it is decidable whether $\cA$ has the finite-image property.
    \end{proof}
    
     We give examples illustrating applications of Theorem~\ref{thm:characterization}.
    
    \newpage
    
    \begin{example}\rm\hfill
        \begin{compactenum}
            \item Let $B$ be an additively locally finite strong bimonoid. In this case Theorem~\ref{thm:characterization}(2) is satisfied. This is the situation of Theorem \ref{cor:image-A-finite-decidable}. Note, however, Theorem \ref{cor:image-A-finite-decidable} does not require $B$ to be computable.
            \item Let $B$ be a strong bimonoid such that, for each $b \in (B \setminus \{\0\})$, the set $\langle b \rangle_\oplus$ is not finite, e.g., the semiring $\bbN$ of natural numbers. In this case Theorem~\ref{thm:characterization}(2) is also satisfied. This situation can be viewed as orthogonal to the situation of Theorem \ref{cor:image-A-finite-decidable}.
            \item Each semiring $B$ belongs to either Item 1 or Item 2, because by Lemma~\ref{lm:monotonic-semiring}, it suffices to decide whether $\1 = \1 \oplus \1$.
            \item The following strong bimonoid combines features of Items 1 and 2, i.e., it contains elements both which have, respectively, which do not have finite additive order, but this is decidable for each element.
            Let $B = (\bbN,\oplus,\cdot,0,1,\leq)$, where, for each $x,y \in \bbN$, we let
            \[x \oplus y = \begin{cases}
                \max\{x,y\} &\text{if  $x \in \{0,1,2\}$ or $y \in \{0,1,2\}$}\\
                x + y &\text{otherwise}\enspace;
            \end{cases}\]
            in addition, $\cdot$ is the usual multiplication, and $\leq$ is the usual order on $\bbN$. Then, for each $b \in B$, we can test whether $b \in \{0,1,2\}$ to decide whether $b$ has finite additive order. Thus, Theorem~\ref{thm:characterization}(2) is satisfied. \qed
        \end{compactenum}
    \end{example}

\section{Further decidability results for wta over past-finite monotonic strong bimonoids} \label{sect:main-result}

In all of this section,  $B$  will be a past-finite monotonic and computable strong bimonoid. We recall that, by Theorem~\ref{thm:inverse_b_recognizable}(1), a $(\Sigma,B)$-wta  $\A$  always has the preimage property.
Hence, by Lemma~\ref{lm:cdwta_finite_image}, the $(\Sigma,B)$-wta  $\A$ has the finite-image property iff  $\sem{\A}$  is a recognizable step mapping iff  $\A$  is crisp-determinizable. Hence, the question arises whether in case  $\A$  has the finite image property, and thus, is crisp-determinizable, we can actually construct an equivalent crisp-deterministic wta. This will be shown to be true in the following result.
Note that then, from the equivalent crisp-deterministic wta, we also know the elements of the finite set $\im(\sem{\A})$  as well as tree automata for the recognizable tree languages $\sem{\A}^{-1}(b)$ (for all  $b \in \im(\sem{\A})$).

\begin{theorem} \label{thm:construction-of-crisp-deterministic-wta}  Let $B$ be past-finite monotonic  and computable. Moreover, let $\A=(Q,\delta,F)$ be a $(\Sigma,B)$-wta given effectively.  If  $\A$ has the finite-image property, then we can construct effectively  a crisp-deterministic  $(\Sigma,B)$-wta $\B$ such that $\sem{\B}=\sem{\A}$.
\end{theorem}

\begin{proof} By Lemma~\ref{lm:decide-useful-state}, it is decidable whether $\A$ has a useful state. If $\A$ does not have a useful state, then $\sem{\A}=\widetilde{\0}$ and we can construct effectively  the desired crisp-deterministic  $(\Sigma,B)$-wta $\B$ in an easy way. 

Otherwise, by Theorem \ref{thm:A'-equivalent-to-A} we may assume that $\A$ is trim.
Then, by  Theorem \ref{cor:characterization-of-crisp-determinizability} and Corollary \ref{cor:small-loops->cA-finite}(1), $\A$ satisfies the assumption of Theorem~\ref{thm:sufficient-conds-ensure-rec-step-map}.  Thus we can use Theorem~\ref{thm:sufficient-conds-ensure-rec-step-map}(2) to construct effectively a crisp-deterministic  $(\Sigma,B)$-wta $\B$ such that $\sem{\B}=\sem{\A}$.
\end{proof}

In the rest of this section we will show that if $B$ is past-finite monotonic and computable, then  we can decide,  for every wta $\A$ given effectively and positive integer $k$, whether the cardinality of the image of $\sem{\A}$ is bounded by $k$,
and in this case we can construct effectively a crisp-deterministic wta which is equivalent to $\A$.

\begin{theorem}  \label{thm:dec-imA-leq-k}
    Let $B$ be  past-finite monotonic and computable. It is decidable, for every  $(\Sigma,B)$-wta $\A$ given effectively and $k \in \N_+$, whether we have $|\im(\sem{\A})| \leq k$. Moreover, in this case we can construct effectively a crisp-deterministic $(\Sigma,B)$-wta $\B$ such that $\sem{\B} = \sem{\A}$.
  \end{theorem}
  \begin{proof}  First, by the proof of Lemma~\ref{lm:decide-useful-state}, we check whether  $\A$  has a useful state.
If  $\A$  has no useful state, we have  $\sem{\A} = \widetilde{0}$, hence our statement holds obviously. Therefore we may assume that
$\A$  has a useful state. By Theorem~\ref{thm:A'-equivalent-to-A}, we may assume that  $\A$  is trim.

Next we check whether  all small loops of $\A$  have weight $\1$ (cf. proof of Theorem~\ref{cor:image-A-finite-decidable}).
If this is not the case, then by Theorem~\ref{cor:characterization-of-crisp-determinizability},  $\im(\sem{\A})$  is infinite and hence  $\A$  is not crisp-determinizable. 
Otherwise,  $\rH(\A)$ is finite by Corollary~\ref{cor:small-loops->cA-finite}(1) and thus $\rC(\A)$ is also finite. (We recall that each monotonic strong bimonoid is one-product free.)
  Using Lemma~\ref{lm:compute-c(A)}, we compute $\rC(\A)$.

We choose an effective enumeration $\xi_0,\xi_1,\ldots$ of $\T_\Sigma$.
Then we run the following Algorithms A and B in parallel.

\

Algorithm A:  Compute  $\sem{\A}(\xi_0), \sem{\A}(\xi_1), \ldots$.
We let this algorithm terminate if we have obtained more than  $k$  different values.

\

Algorithm B:  For each  $b \in \rC(A)$, we run the following subalgorithm  $\mathrm{Alg}(b)$.\\
$\mathrm{Alg}(b)$:  We compute the $(\Sigma,\N)$-wta  $\A'_b$  as in the proof of Theorem~\ref{thm:sufficient-conds-ensure-rec-step-map}.
Now we proceed as in the algorithm described in the proof of Theorem~\ref{thm:sufficient-conds-ensure-rec-step-map}(2). Successively for $i = 0,1,\ldots$, 
\begin{compactenum}
\item[(a)] we compute and store a finite-state $\Sigma$-tree automaton  for the language  $L_{b,i} = \sem{\A'_b}^{-1}(i)$ (cf. Lemma~\ref{cor:inverse_n_recognizable}(2)), and 
\item[(b)] we compute and store  the element  $ib$.
\end{compactenum}
We let  $\mathrm{Alg}(b)$  terminate, if for some  $i \in \N$,  we have $\bigcup_{j \in [0,i]} L_{b,j} = \T_\Sigma$ or  $ib = jb$  for some  $j < i$.
Clearly, both equalities just given are decidable.
We let Algorithm B terminate if, for each $b \in \rC(A)$, the algorithm  $\mathrm{Alg}(b)$  terminates.

\

Next we show that this decision procedure terminates
and thus we can decide whether $|\im(\sem{\A})| \leq k$.
Clearly, if  Algorithm A terminates, we have  $k ~<~|\im(\sem{\A})|$.

Therefore let us assume that Algorithm A does not terminate.
In this case, we have $|\im(\sem{\A})| \leq k$, and we have to show that Algorithm B terminates.
By Theorem~\ref{cor:characterization-of-crisp-determinizability}, for each $b \in \rC(\A)$, the mapping $f_{\A,b}$ is bounded or $b$ has finite additive order. Hence, the assumptions of Theorem~\ref{thm:sufficient-conds-ensure-rec-step-map} are satisfied. As shown in the proof of Theorem~\ref{thm:sufficient-conds-ensure-rec-step-map}(2), for each $b \in \rC(\A)$, the algorithm $\mathrm{Alg}(b)$ terminates. Hence, Algorithm B terminates.

It follows Algorithm A or Algorithm B terminates. If Algorithm A terminates while Algorithm B is running, we know that  $k < |\im(\sem{\A})|$. Now assume Algorithm B terminates while Algorithm A is still running. Then by the proof of Theorem~\ref{thm:sufficient-conds-ensure-rec-step-map}(2),
we can compute effectively a crisp-deterministic wta  $\B$ equivalent to $\A$.
Then we decide whether  $|\im(\sem{\B})| \leq k$. 
\end{proof}

Besides the proof of Theorem \ref{thm:dec-imA-leq-k} we also show Algorithm~\ref{alg:decides-if-bounded-by-k} which implements the decision procedure in that proof
as a pseudo code. We explain the implementation details as follows.   

We assume that we have an enumeration $\xi_0,\xi_1, \ldots$ of $\T_\Sigma$.

As first step (lines \ref{line:begin-test-useful-state}--\ref{line:end-test-useful-state}), Algorithm \ref{alg:decides-if-bounded-by-k} tests whether $\A$ has a useful state as in the proof of Lemma~\ref{lm:decide-useful-state}. If the answer is no, then we have $\sem{\A}=\widetilde{\0}$. Hence  Algorithm \ref{alg:decides-if-bounded-by-k} constructs effectively a crisp-deterministic $(\Sigma,B)$-wta $\B$ with $\sem{\B}=\widetilde{\0}$ and terminates with output ``yes''. Otherwise,  Algorithm \ref{alg:decides-if-bounded-by-k} trims $\A$ according to Theorem \ref{thm:A'-equivalent-to-A}.

As second step (lines \ref{line:begin-test-small-loops}--\ref{line:end-test-small-loops}), Algorithm \ref{alg:decides-if-bounded-by-k} tests whether every small loop of $\A$ has weight $\1$. This can be done by the same argument as in the proof of Theorem~\ref{cor:image-A-finite-decidable}). If this is not true, then $\im(\sem{\A})$ is infinite (by Theorem~\ref{cor:characterization-of-crisp-determinizability}) and Algorithm~\ref{alg:decides-if-bounded-by-k} terminates with output ``no''. Otherwise, Algorithm~\ref{alg:decides-if-bounded-by-k} continues and, by Corollary~\ref{cor:small-loops->cA-finite}(1), the set $\rH(\A)$ is finite, and hence also $\C(\A)$ is finite. 

As third step (lines \ref{line:compute-C(A)}--\ref{line:end-init-b-in-C(A)}), Algorithm~\ref{alg:decides-if-bounded-by-k} computes the set $\rC(\A) \subseteq B$; by Lemma~\ref{lm:compute-c(A)} this is possible. Moreover, for each $b \in \rC(\A)$, a Boolean value $\mathrm{flag}(b)$, a set $S_b \subseteq B$, and a set $U_b \subseteq \T_\Sigma$ are initialized. The intuition for $\mathrm{flag}$, $S_b$, and $U_b$ are as follows.
  \[
    \mathrm{flag}(b) = \begin{cases}
      \mathrm{true} & \text{if $f_{\A,b}$ is bounded or $\langle b\rangle_\oplus $  is finite}\\
      \mathrm{false} & \text{ otherwise}
      \end{cases} 
    \]
    In the set $S_b$, Algorithm~\ref{alg:decides-if-bounded-by-k} collects multiples of $b$, i.e., $S_b \subseteq \{ib \mid i \in \N\}$. Moreover, in the set $U_b$, Algorithm~\ref{alg:decides-if-bounded-by-k} collects all those trees $\xi$ for which the multiplicity $f_{\A,b}(\xi)$ is $i$ for some $i \in \N$, i.e.,
    \[
U_b \subseteq \bigcup_{i \in \N} L_{b,i} \ \text{ where } \ L_{b,i} = \{\xi \in \T_\Sigma \mid f_{\A,b}(\xi) = i\}\enspace.
\]
We note that Algorithm~\ref{alg:decides-if-bounded-by-k} does not compute the set $U_b$ itself, but a finite representation of $U_b$ in the form of a finite-state $\Sigma$-tree automaton (cf. the proof of Theorem~\ref{thm:sufficient-conds-ensure-rec-step-map}).

  As fourth step (\ref{line:begin-dec-alg}--\ref{line:end-dec-alg}),  Algorithm~\ref{alg:decides-if-bounded-by-k} runs the following two algorithms in parallel. Here, however, we explain these two algorithms separately.

  \begin{itemize}
\item Algorithm A: \sloppy Algorithm A computes  $\sem{\A}(\xi_0), \sem{\A}(\xi_1), \ldots$  and collects these values into a set $W$ (line \ref{line:collect-weight}).  To compute these values, Algorithm A uses the given enumeration $\xi_0,\xi_1, \ldots$ of $\T_\Sigma$. Algorithm A terminates if more than $k$ different values were obtained (line \ref{line:k+1-different-values}). Then the answer to the decision problem is ``no''.

\item Algorithm B:  For each  $b \in \C(\A)$, Algorithm B computes simultaneously the sequences $L_{b,0}$, $L_{b,1}, \ldots$ (line \ref{line:comonstruct-tree-automata}) and $0b,1b, \ldots$ (line \ref{line:compute-multiples-of-b}). 
  Algorithm B terminates, if $\mathrm{flag}(b)=\mathrm{true}$ for each $b \in \rC(\A)$ (lines \ref{line:set-flag-fAb-bounded} and \ref{line:set-flag-b-has-fin-add-ord}), i.e.,
for each $b \in \rC(\A)$ there exists an $i \in \N$ such that $\bigcup_{j \in [0,i]} L_{b,j} = \T_\Sigma$ or  $ib = jb$  for some  $j < i$. Clearly, both tests ($U_b = \T_\Sigma$ and $ib \in S_b$) are decidable.
\end{itemize}
The parallel running of Algorithms A and B terminates,  and after termination the assumptions of Theorem~\ref{thm:sufficient-conds-ensure-rec-step-map} are satisfied (cf. the proof of Theorem \ref{thm:dec-imA-leq-k}).
Then, by Theorem~\ref{thm:sufficient-conds-ensure-rec-step-map}(2), Algorithm~\ref{alg:decides-if-bounded-by-k} can compute effectively a crisp-deterministic wta  $\B$  equivalent to $\A$ and it decides whether $|\im(\sem{\B})| \leq k$.
This finished the explanation of Algorithm~\ref{alg:decides-if-bounded-by-k}.

\

\begin{algorithm}\label{alg:decides-if-bounded-by-k}
    \small
    \KwIn{a $(\Sigma,B)$-wta $\A=(Q,\delta,F)$ given effectively,\;
    \hspace*{10mm} an effective enumeration $\xi_0,\xi_1,\ldots$ of $\T_\Sigma$, and $k \in \N_+$
    }
    \KwOut{"yes" if $|\im(\sem{\A})| \leq k$ and "no" otherwise}
    \BlankLine
    \Var $i \in \N$, $b \in B$, $W \subseteq B$, $\xi \in \T_\Sigma$\;
    \hspace*{15mm} $\mathrm{flag}: \rC(\A) \rightarrow \{\mathrm{true}, \mathrm{false}\}$\;
    \hspace*{15mm} family $\big(U_b \subseteq \T_\Sigma \mid b \in \rC(\A)\big)$ \;
    \hspace*{15mm} family $\big(S_b \subseteq B \mid b \in \rC(\A)\big)$ \;
    \BlankLine
    \If(\Comment*[f]{\textrm{\% cf. proof of Lemma~\ref{lm:decide-useful-state}}}){$\neg(\A \mathrm{\ has\ a\ useful\ state})$ \label{line:begin-test-useful-state}}{
        construct effectively a crisp-deterministic $(\Sigma,B)$-wta $\B$ such that $\sem{\B} = \0$\;
        \Return{$"\mathrm{yes}"$}
    } \label{line:end-test-useful-state} 
    trim $\A$ \Comment*[r]{\textrm{\% cf. Theorem \ref{thm:A'-equivalent-to-A}}}
    \If(\Comment*[f]{\textrm{\% cf. proof of Theorem~\ref{cor:image-A-finite-decidable}}}){$\neg(\mathrm{all\ small\ loops\ of\ } \A \mathrm{\ have\ weight\ } \1)$ \label{line:begin-test-small-loops}}{
      \Return{$"\mathrm{no}"$} \Comment*[r]{\textrm{\% $\im(\sem{\A})$ is infinite, cf. Theorem~\ref{cor:characterization-of-crisp-determinizability}}}
    }\label{line:end-test-small-loops}
    \Comment*[r]{\textrm{\% $\rH(\A)$ is finite, cf. Corollary~\ref{cor:small-loops->cA-finite}(1)}} 	
    compute $\rC(\A)$ \label{line:compute-C(A)}\Comment*[r]{\textrm{\% cf. Lemma~\ref{lm:compute-c(A)}}}
    \ForEach{$b \in \rC(\A)$}{
        $\mathrm{flag}(b) \leftarrow \mathrm{false}$, 
        $U_b \leftarrow \emptyset$, and $S_b \leftarrow \emptyset$\;
        construct effectively the $(\Sigma,\N)$-wta $\A'_b$ \Comment*[r]{\textrm{\% cf. proof of Theorem~\ref{thm:sufficient-conds-ensure-rec-step-map}(2)}}
      } \label{line:end-init-b-in-C(A)}
      \BlankLine
    $i \leftarrow 0$ and $W \leftarrow \emptyset$ \label{line:begin-dec-alg}\;
    \While{$\mathrm{true}$}{
        $\xi \leftarrow \xi_i$ \Comment*[r]{\textrm{\% query the next tree}}
        $W \leftarrow W \cup \{\sem{\A}(\xi)\}$ \label{line:collect-weight}\;
        \lIf{$|W| > k$ \label{line:k+1-different-values}}{\Return{$"\mathrm{no}"$}}
        \eIf{$\mathrm{flag}^{-1}(\mathrm{false}) \neq \emptyset$}{
            \ForEach{$b \in \mathrm{flag}^{-1}(\mathrm{false})$}{
                $U_b \leftarrow U_b \cup \sem{\A'_b}^{-1}(i)$ \label{line:comonstruct-tree-automata}\Comment*[r]{\textrm{\% cf. Lemma \ref{cor:inverse_n_recognizable}(2)}} 
                \eIf(\Comment*[f]{\textrm{\% check whether $f_{\A,b}$ is bounded}}){
                    $U_b = \T_\Sigma$ \label{line:testing-boundedness}}
                    {$\mathrm{flag}(b) \leftarrow \mathrm{true}$ \label{line:set-flag-fAb-bounded}}
                    {
                        \eIf(\Comment*[f]{\textrm{\% check whether $b$ has finite additive order}})
                        {$ib \in S_b$\label{line:testing-finiteness}}
                        {$\mathrm{flag}(b) \leftarrow \mathrm{true}$ \label{line:set-flag-b-has-fin-add-ord}}
                        {$S_b \leftarrow S_b \cup \{ib\}$ \label{line:compute-multiples-of-b}}
                    }                      
            }
            $i \leftarrow i + 1$\;
        }(\Comment*[f]{\textrm{\% conditions of Theorem~\ref{thm:sufficient-conds-ensure-rec-step-map}(2) hold}}){
            construct effectively the crisp-deterministic $(\Sigma,B)$-wta $\B$ such that $\sem{\B}=\sem{\A}$ \;
            \leIf{$|\im(\sem{\B})| \leq k$}{\Return{$"\mathrm{yes}"$}}{\Return{$"\mathrm{no}"$}}
        }
        
    } \label{line:end-dec-alg}
    \caption{\label{alg:dec-imA-leq-k} Deciding whether $\im(\sem{\A}) \leq k$}
\end{algorithm}

To draw conclusions from Theorem \ref{thm:dec-imA-leq-k}, we need the following notions.

Let $r$ be a $(\Sigma,B)$-weighted tree language, $E \subseteq B$ be a finite set, and $b \in B$.
We say that $r$ is a {\em $(\Sigma,B,E)$-recognizable step mapping} (for short: $E$-recognizable step mapping) if $r$ is a recognizable step mapping and $\im(r) = E$.  Moreover, we say that $r$ is {\em constant} if $r=\widetilde{b}$ for some $b \in B$.

Now we show that, for arbitrary  past-finite monotonic and computable strong bimonoid $B$,  $(\Sigma,B)$-wta $\A$ given effectively, finite subset $E\subseteq B$, and $b\in B$, it is decidable whether $\sem{\A}$ is an $E$-recognizable step mapping, and whether $\sem{\A} = \widetilde{b}$
(for the definition of $\widetilde{b}$, see Section \ref{sect:wtl}).
Moreover, we revisit the decidability results of Borchardt (\cite[Sect.~6]{bor04}) concerning the constant problem and the constant-on-its-support problem and prove them for a larger class of algebras.

\begin{corollary} \rm \label{cor:decidability-results}
Let $B$ be past-finite monotonic and computable. Then, for each $(\Sigma,B)$-wta $\A$ given effectively, the following questions are decidable:
\begin{compactenum}
    \item[(a)] Given a finite subset $E \subseteq B$, is the weighted tree language $\sem{\A}$ an $E$-recognizable step mapping?
    \item[(b)] Given $b \in B$, is $\sem{\A} = \widetilde{b}$?
    \item[(c)] Is $\sem{\A}$ constant?
    \item[(d)] Is $\sem{\A}$ a recognizable one-step mapping?
\end{compactenum}
\end{corollary}

\begin{proof} First, for each case (a)-(d), we define a number $k$. 
    In case (a), let $k = |E|$. In cases (b) and (c), we put $k = 1$. In case (d), we set $k=2$. 
Let us decide whether $|\im(\sem{\A})|\le k$ (cf. Theorem \ref{thm:dec-imA-leq-k}).
If this is not the case, the answer to the respective question is "no". Otherwise, a crisp-deterministic $(\Sigma,B)$-wta $\B$ is given such that $\sem{\B}=\sem{\A}$. Now for this $\B$ we can decide the questions of (a) -- (d).
\end{proof}

\section{A comparison of the concepts cost-finiteness and having the finite-image property}

Let $B$ be a strong bimonoid. A $(\Sigma,B)$-wta $\A$ is called {\em cost-finite} if the set
\[\rH(\A)^{F \neq \0} = \{\wt(\xi,\rho) \mid \xi \in \T_\Sigma, \rho \in \R_\A(\xi), \text{ and } F_{\rho(\varepsilon)} \neq \0\}\enspace,\]
is finite (cf. \cite{borfulgazmal05}). We note that in \cite{borfulgazmal05} $\rH(\A)^{F \neq \0}$ is denoted by $\mathrm{c}(\A)$.

In \cite[Thm. 44]{borfulgazmal05}, a characterization of cost-finiteness of trim wta was given
for finitely factorizing monotonic semirings. In this section, we wish to give such a
characterization for monotonic strong bimonoids. As a corollary, we obtain that cost-finiteness is decidable for  wta over monotonic strong bimonoids with effective tests for  $\0$  and  $\1$. Note however, that the weights of transitions of wta in \cite{borfulgazmal05} have a more general structure
than the weights considered here (they are polynomials instead of monomials).
Therefore, Lemma \ref{prop:app}  and Corollary \ref{cor:past-finite-cost-finite} cannot be considered as a generalization of the
corresponding results  \cite[Thm. 44, 46]{borfulgazmal05}. Finally, we compare  the concept of cost-finiteness \cite{borfulgazmal05} and the concept of having the finite-image property (cf. Corollary~\ref{cor:concept-comparision}).

\begin{lemma}\rm \label{prop:app}  Let  $B$  be a monotonic strong bimonoid, and let  $\A$  be trim. Then  $\A$  is cost-finite if and only if small loops of  $\A$
have weight~$\1$.
\end{lemma}

\begin{proof} If small loops of  $\A$  have weight  $\1$, then $\A$ is cost-finite by Corollary \ref{cor:small-loops->cA-finite}(1).

Now assume  $\A$  has a loop $\rho$ on a context $c$ with $\1\prec \wt(c,\rho)$.
Then as in \eqref{eq:prec-sequence} of the  proof of Theorem~\ref{cor:characterization-of-crisp-determinizability},  we can produce an infinite set
of weights of runs on trees.  Hence  $\rH(\A)^{F \neq \0}$  is infinite.
\end{proof}

As a consequence, we obtain that cost-finiteness of wta is decidable.

\begin{corollary}\rm \label{cor:past-finite-cost-finite}  Let  $B$  be monotonic and have effective tests
for  $\0$  and  $\1$.  Then it is decidable whether an arbitrary
$(\Sigma,B)$-wta is cost-finite.
\end{corollary}

\begin{proof} Let $\A=(Q,\delta,F)$ be an arbitrary $(\Sigma,B)$-wta. We first check whether $\A$ has a useful state as in the proof of Lemma~\ref{lm:decide-useful-state}. If $\A$ has no useful state, then $\rH(\A)^{F \neq \0} = \emptyset$, and thus, $\A$ is obviously cost-finite. Otherwise, by Theorem~\ref{thm:A'-equivalent-to-A}, we may assume that $\A$ is trim. Moreover, by Lemma~\ref{prop:app}, $\A$ is cost-finite if and only if small loops of $\A$ have weight~$\1$. The latter property is decidable by the proof of Theorem~\ref{cor:image-A-finite-decidable}.
\end{proof}

Lastly, we remark that cost-finiteness and having the finite-image property coincide for wta over additively locally finite and past-finite monotonic strong bimonoids.

\begin{corollary} \rm \label{cor:concept-comparision}
Let $B$ be monotonic and $\A=(Q,\delta,F)$ be a $(\Sigma,B)$-wta. 
\begin{compactenum}
\item If $B$ is additively locally finite and $\A$ is cost-finite, then $\A$ has the finite-image property.
\item If $B$ is past-finite and $\A$ has the finite-image property, then $\A$ is cost-finite.
\end{compactenum}
\end{corollary}

\begin{proof} 
We first prove Statement 1. Since $\A$ is cost-finite, the set $\rC(\A)$ is finite. Moreover, we note that $\im(\sem{\A})$ is contained in $\langle \rC(\A)\rangle_\oplus$. By the assumption of Statement~1, the latter set is finite. 

We prove Statement 2: Since $B$ is monotonic, for every $\xi \in \T_\Sigma$, $q \in Q$ with $F_q \neq \0$, and $\rho \in \R_\A(q,\xi)$, 
we have $\wt(\xi,\rho)\preceq \wt(\xi,\rho)\otimes F_q \preceq \sem{\A}(\xi)$. Hence, \[\rH(\A)^{F \neq \0}\subseteq \bigcup_{b\in \im(\sem{\A})} \past(b),\]
and the set on the right hand side of the inclusion is finite because $B$ is past-finite and $\A$ has the finite-image property. Thus $\A$ is cost-finite.
\end{proof}

We note that, in general, the implication of Corollary \ref{cor:concept-comparision}(2) does not hold if the condition past-finite is dropped. In fact, the wta $\cA$ in Example~\ref{ex:cA-infin-imA-fin} has the finite-image property, but it is not cost-finite.

In the following example we show that past-finiteness of $B$ and $\A$ being cost-finite do not imply that $\A$ has the finite-image property. In fact, we can give a $(\Sigma,\N)$-wta $\A$ such that $\A$ is cost-finite but $\A$ does not have the finite-image property as follows. (Note that $\N$ is a past-finite monotonic.)

\begin{example} \label{ex:cA-fin-imA-infin} \rm
 Let $\Sigma=\{\gamma^{(1)},e^{(0)}\}$. We consider the $(\Sigma,\N)$-wta $\A=(\{p,q,r\},\delta,F)$ with $\delta_0(\varepsilon,e,p)=\delta_1(p,\gamma,q)=\delta_1(q,\gamma,p)=\delta_1(p,\gamma,r)=\delta_1(r,\gamma,p)=F_p = 1$, all other transitions have weight $0$, and $F_q = F_r = 0$. 

Clearly, $\rH(\A)^{F \neq \0}=\{1\}$, i.e., $\A$ is cost-finite. Moreover, for every $n \in \N$ and $\gamma^{2n}(e) \in \T_\Sigma$, we have $\sem{\A}\big(\gamma^{2n}(e)\big)=2^n$. Since $\{2^n \mid n \in \N\} \subseteq \im(\sem{\A})$, the wta $\A$ does not have the finite-image property.

We remark that, though $\N$ is commutative and one-product-free, small loops of $\A$ have weight 1,
the conditions of Corollary \ref{cor:small-loops->cA-finite}(2) do not hold. In fact, $\rC(\A)=\{0,1\}$ and neither $f_{\A,1}$ is bounded nor 1 has finite additive order. \qed
\end{example}

\section{Results for weighted string automata} \label{sect:string-case}

Let $\Delta$ be an alphabet. A \emph{weighted string automaton (over $\Delta$ and  $B$)} 
(for short: $(\Delta,B)$-wsa)   \cite{sch61,eil74} is a tuple
$\A = (Q,I,\delta,F)$, where $Q$ is a finite set of
\emph{states}, $I: Q\to B$ is the \emph{initial weight mapping}, $\delta: Q\times \Delta \times Q\to B$ is the \emph{transition
mapping}, and  $F: Q \to B$ is the \emph{final weight mapping}. For each $q \in Q$, we abbreviate $I(q)$ and $F(q)$ by $I_q$ and $F_q$, respectively.

We define the run semantics for $\A$ as follows.  Let $w=a_1\cdots a_n$ be a string in $\Delta^*$ with $n \in \mathbb{N}$ and $a_i \in \Delta$ for each $i \in [n]$. A {\em run of $\A$ on $w$} is a string $\rho= q_0\cdots q_n$ in $Q^{n+1}$. 
The {\em weight of $\rho$ for $w$}, denoted by $\wt(w,\rho)$, is the element of $B$ defined by 
\[\wt(w,\rho) = I_{q_0}\otimes \delta(q_0,a_1,q_1)\otimes \ldots \otimes \delta(q_{n-1},a_n,q_n) \otimes F_{q_n}\enspace.\]
Then the {\em run semantics of $\A$} is the weighted language
$\sem{\A}:\Delta^* \to B$ defined by 
\[\sem{\A}(w)=\bigoplus_{\rho\in Q^{|w|+1}}\wt(w,\rho)\] 
for every $w\in \Delta^*$. In particular, $\sem{\A}(\varepsilon) = \bigoplus_{q \in Q} I_{q}\otimes F_{q}$. A weighted language $r: \Delta^* \rightarrow B$ is \emph{run-recognizable} if there exists a $(\Delta,B)$-wsa $\A$ such that $r = \sem{\A}$.

In \cite[p.~324]{fulvog09} it is shown that, for each semiring $B$, the concept of  $(\Delta,B)$-wsa
and the concept of $(\Sigma,B)$-wta where $\Sigma$ is a string ranked alphabet are essentially the same. A string ranked alphabet is a ranked alphabet $\Sigma$ for which $\Sigma = \Sigma^{(0)} \cup \Sigma^{(1)}$ and $|\Sigma^{(0)}| =1$. In fact, for each $(\Delta,B)$-wsa $\A$ there exists a string ranked alphabet $\Sigma$, a bijection $\tree:\Delta^*\to \T_\Sigma$, and a $(\Sigma,B)$-wta $\B$ such that $\sem{\A}(w)=\sem{\B}(\tree(w))$ for each $w\in \Delta^*$. The inverse of this statement also holds, and the proof of both directions also works if $B$ is a strong bimonoid.

Since each string ranked alphabet is a particular ranked alphabet, each of our results for wta also holds for wsa with run semantics.  Moreover, the wta in Examples \ref{ex:commutative-SB-not-crisp-determinizable}, \ref{ex:cA-fin-imA-infin}, and \ref{ex:cA-infin-imA-fin} are examples for wsa with run semantics because in each of these examples, the ranked alphabet of the wta is a string ranked alphabet.

\

%
\textbf{Acknowledgments.} The authors had obtained Theorem~\ref{thm:past-fin-mon-semiring-fin-im-dec} for semirings. They are thankful to Uli Fahrenberg for a question which prompted the extension of this semiring-result
to left or right distributive strong bimonoids.

\bibliographystyle{alpha}
\bibliography{../crisp19-bib}

\newcommand{\etalchar}[1]{$^{#1}$}
\begin{thebibliography}{DGMM11}

\bibitem[AB87]{aleboz87}
A.~Alexandrakis and S.~Bozapalidis.
\newblock Weighted grammars and {K}leene's theorem.
\newblock {\em Inform. Process. Lett.}, 24(1):1--4, 1987.

\bibitem[BFGM05]{borfulgazmal05}
B.~Borchardt, Z.~F{\"u}l{\"o}p, Z.~Gazdag, and A.~Maletti.
\newblock Bounds for tree automata with polynomial costs.
\newblock {\em J. Autom., Lang. Comb.}, 10:107--157, 2005.

\bibitem[BM{\v{S}}{\etalchar{+}}06]{bormalsestepvog06}
B.~Borchardt, A.~Maletti, B.~{\v{S}}e\v{s}elja, A.~Tepav\v{c}evic, and
  H.~Vogler.
\newblock Cut sets as recognizable tree languages.
\newblock {\em Fuzzy Sets and Systems}, 157:1560--1571, 2006.

\bibitem[Bor04]{bor04}
B.~Borchardt.
\newblock A pumping lemma and decidability problems for recognizable tree
  series.
\newblock {\em Acta Cybernet.}, 16(4):509--544, 2004.

\bibitem[Bor05]{bor04b}
B.~Borchardt.
\newblock {\em The Theory of Recognizable Tree Series}.
\newblock Verlag f{\"u}r {W}issenschaft und {F}orschung, 2005.
\newblock (Ph.D. thesis, 2004, TU Dresden, Germany).

\bibitem[BR82]{berreu82}
J.~Berstel and C.~Reutenauer.
\newblock Recognizable formal power series on trees.
\newblock {\em Theoret. Comput. Sci.}, 18(2):115--148, 1982.

\bibitem[BR88]{berreu88}
J.~Berstel and Ch. Reutenauer.
\newblock {\em Rational Series and Their Languages}, volume~12 of {\em EATCS
  Monographs on Theoretical Computer Science}.
\newblock Springer-Verlag, 1988.

\bibitem[CDG{\etalchar{+}}08]{tata08}
H.~Comon, M.~Dauchet, R.~Gilleron, C.~L\"oding, F.~Jacquemard, D.~Lugiez,
  S.~Tison, and M.~Tommasi.
\newblock Tree automata techniques and applications.
\newblock Available on: {\tt http://tata.gforge.inria.fr}, 2008.

\bibitem[CDIV10]{cirdroignvog10}
M.~\'Ciri\'c, M.~Droste, J.~Ignjatovi\'c, and H.~Vogler.
\newblock Determinization of weighted finite automata over strong bimonoids.
\newblock {\em Inform. Sci.}, 180(18):3479--3520, 2010.

\bibitem[DFKV20]{drofulkosvog20}
M.~Droste, Z.~F{\"u}l{\"o}p, D.~K{\'o}sz{\'o}, and H.~Vogler.
\newblock Crisp-determinization of weighted tree automata over additively
  locally finite and past-finite monotonic strong bimonoids is decidable.
\newblock In G.~Jir\'askov\'a and G.~Pighizzini, editors, {\em Descriptional
  Complexity of Formal Systems (DCFS 2020)}, volume 12442 of {\em LNCS}, pages
  39--51. Springer Nature Switzerland, 2020.

\bibitem[DG05]{drogas05}
M.~Droste and P.~Gastin.
\newblock Weighted automata and weighted logics.
\newblock In L.~Caires, G.~F. Italiano, L.~Monteiro, C.~Palamidessi, and
  M.~Yung, editors, {\em Automata, Languages and Programming -- 32nd Int.
  Colloquium, ICALP 2005}, volume 3580 of {\em LNCS}, pages 513--525.
  Springer-Verlag, 2005.

\bibitem[DGMM11]{drogoemaemei11}
M.~Droste, D.~G{\"o}tze, S.~M{\"a}rcker, and I.~Meinecke.
\newblock Weighted tree automata over valuation monoids and their
  characterizations by weighted logics.
\newblock In W.~Kuich and G.~Rahonis, editors, {\em Algebraic Foundations in
  Computer Science}, volume 7020 of {\em LNCS}, pages 30--55. Springer, 2011.

\bibitem[DHV15]{droheuvog15}
M.~Droste, D.~Heusel, and H.~Vogler.
\newblock Weighted unranked tree automata over tree valuation monoids and their
  characterization by weighted logics.
\newblock In A.~Maletti, editor, {\em Algebraic Informatics (CAI 2015)}, volume
  9270 of {\em LNCS}, pages 90--102. Springer, 2015.

\bibitem[DK09]{drokui09}
M.~Droste and W.~Kuich.
\newblock Semirings and formal power series.
\newblock In M.~Droste, W.~Kuich, and H.~Vogler, editors, {\em Handbook of
  Weighted Automata}, Monographs in Theoretical Computer Science. An EATCS
  Series, chapter~1, pages 3--28. Springer-Verlag, 2009.

\bibitem[DKV09]{drokuivog09}
M.~Droste, W.~Kuich, and H.~Vogler, editors.
\newblock {\em Handbook of Weighted Automata}.
\newblock EATCS Monographs in Theoretical Computer Science. Springer-Verlag,
  2009.

\bibitem[DM12]{dromei12}
M.~Droste and I.~Meinecke.
\newblock Weighted automata and weighted {MSO} logics for average and long-time
  behaviors.
\newblock {\em Information and Computation}, 220--221:44--59, 2012.

\bibitem[DP16]{droper16}
M.~Droste and V.~Perevoshchikov.
\newblock Multi-weighted automata and {MSO} logic.
\newblock {\em Theory Comput. Syst.}, 59:231--261, 2016.

\bibitem[DSV10]{drostuvog10}
M.~Droste, T.~St\"uber, and H.~Vogler.
\newblock Weighted finite automata over strong bimonoids.
\newblock {\em Inform. Sci.}, 180(1):156--166, 2010.

\bibitem[DV06]{drovog06}
M.~Droste and H.~Vogler.
\newblock Weighted tree automata and weighted logics.
\newblock {\em Theoret. Comput. Sci.}, 366:228--247, 2006.

\bibitem[DV12]{drovog12}
M.~Droste and H.~Vogler.
\newblock Weighted automata and multi-valued logics over arbitrary bounded
  lattices.
\newblock {\em Theoret. Comput. Sci.}, 418:14--36, 2012.

\bibitem[Eil74]{eil74}
S.~Eilenberg.
\newblock {\em Automata, Languages, and Machines -- Volume A}, volume~59 of
  {\em Pure and Applied Mathematics}.
\newblock Academic Press, 1974.

\bibitem[{\'E}K03]{esikui03}
Z.~{\'E}sik and W.~Kuich.
\newblock Formal tree series.
\newblock {\em J. Autom. Lang. Comb.}, 8(2):219--285, 2003.

\bibitem[{\'E}L07]{esiliu07}
Z.~{\'E}sik and G.~Liu.
\newblock Fuzzy tree automata.
\newblock {\em Fuzzy Sets and Syst.}, 158:1450--1460, 2007.

\bibitem[Eng75]{eng75-15}
J.~Engelfriet.
\newblock Tree automata and tree grammars.
\newblock Technical Report {DAIMI FN-10}, Inst. of Mathematics, University of
  Aarhus, Department of Computer Science, Denmark, 1975.
\newblock see also: arXiv:1510.02036v1 [cs.FL] 7 Oct 2015.

\bibitem[FKV21]{fulkosvog19}
Z.~F{\"u}l{\"o}p, D.~K{\'o}sz{\'o}, and H.~Vogler.
\newblock Crisp-determinization of weighted tree automata over strong
  bimonoids.
\newblock {\em Discrete Mathematics \& Theoretical Computer Science}, 23(1),
  2021.

\bibitem[FMV09]{fulmalvog09}
Z.~F{\"u}l{\"o}p, A.~Maletti, and H.~Vogler.
\newblock A {K}leene theorem for weighted tree automata over distributive
  multioperator monoids.
\newblock {\em Theory Comput. Syst.}, 44:455--499, 2009.

\bibitem[FSV12]{fulstuvog12}
Z.~F\"ul\"op, T.~St\"uber, and H.~Vogler.
\newblock A {B}\"uchi-like theorem for weighted tree automata over
  multioperator monoids.
\newblock {\em Theory Comput. Syst.}, 50(2):241--278, 2012.
\newblock published online 28. October 2010, DOI 10.1007/s00224-010-9296-1.

\bibitem[FV09]{fulvog09}
Z.~F{\"u}l{\"o}p and H.~Vogler.
\newblock Weighted tree automata and tree transducers.
\newblock In M.~Droste, W.~Kuich, and H.~Vogler, editors, {\em Handbook of
  Weighted Automata}, chapter~9, pages 313--403. Springer-Verlag, 2009.

\bibitem[GS84]{gecste84}
F.~G{\'e}cseg and M.~Steinby.
\newblock {\em Tree Automata}.
\newblock Akad{\'e}miai Kiad{\'o}, Budapest, 1984.
\newblock see also: arXiv:1509.06233v1 [cs.FL] 21 Sep 2015.

\bibitem[Har78]{har78}
M.A. Harrison.
\newblock {\em Introduction to Formal Language Theory}.
\newblock Addison-Wesley, 1978.

\bibitem[HMU07]{hopmotull07}
J.E. Hopcroft, R.~Motawi, and J.D. Ullman.
\newblock {\em Introduction to Automata Theory, Languages, and Computation}.
\newblock Pearson, Addison-Wesley, 2007.

\bibitem[IF75]{inafuk75}
Y.~Inagaki and T.~Fukumura.
\newblock On the description of fuzzy meaning of context-free languages.
\newblock pages 301--328. Academic Press, New York, 1975.

\bibitem[KLMP04]{klilommaipri04}
I.~Klimann, S.~Lombardy, J.~Mairesse, and C.~Prieur.
\newblock Deciding unambiguity and sequentiality from a finitely ambiguous
  max-plus automaton.
\newblock {\em Theor. Comput. Sci.}, 327(3):349--373, 2004.

\bibitem[KS86]{kuisal86}
W.~Kuich and A.~Salomaa.
\newblock {\em Semirings, Automata, Languages}, volume~5 of {\em EATCS
  Monographs in Theoretical Computer Science}.
\newblock Springer-Verlag, 1986.

\bibitem[Kui97]{kui97}
W.~Kuich.
\newblock Semirings and formal power series: Their relevance to formal
  languages and automata.
\newblock In G.~Rozenberg and A.~Salomaa, editors, {\em Handbook of Formal
  Languages}, volume~1, chapter~9, pages 609--677. Springer-Verlag, 1997.

\bibitem[Kui99]{kui99}
W.~Kuich.
\newblock Linear systems of equations and automata on distributive
  multioperator monoids.
\newblock In {\em Contributions to General Algebra 12 - Proceedings of the 58th
  Workshop on General Algebra ``58. Arbeitstagung Allgemeine Algebra'', Vienna
  University of Technology. June 3-6, 1999}, pages 1--10. Verlag Johannes Heyn,
  1999.

\bibitem[LB99]{lou99}
O.~Louscou-Bozapalidou.
\newblock Some remarks on recognizable tree series.
\newblock {\em Intern J. Computer Math.}, 70:649--655, 1999.

\bibitem[MS77]{mansim77}
A.~Mandel and I.~Simon.
\newblock On finite semigroups of matrices.
\newblock {\em Theoret. Comput. Sci.}, 5:101--111, 1977.

\bibitem[Rad10]{rad10}
D.~Radovanovi\'c.
\newblock Weighted tree automata over strong bimonoids.
\newblock {\em Novi Sad J. Math.}, 40(3):89--108, 2010.

\bibitem[Rah09]{rah09}
G.~Rahonis.
\newblock Fuzzy languages.
\newblock In M.~Droste, W.~Kuich, and H.~Vogler, editors, {\em Handbook of
  Weighted Automata}, chapter~12, pages 481--517. Springer-Verlag, 2009.

\bibitem[RR85]{resreu85}
A.~Restivo and C.~Reutenauer.
\newblock Rational languages and the burnside problem.
\newblock {\em Theoret. Comput. Sci.}, 40:13--30, 1985.

\bibitem[Sak09]{sak09}
J.~Sakarovitch.
\newblock {\em Elements of Automata Theory}.
\newblock Cambridge University Press, 2009.

\bibitem[Sch61]{sch61}
M.P. Sch{\"u}tzenberger.
\newblock On the definition of a family of automata.
\newblock {\em Inf. and Control}, 4:245--270, 1961.

\bibitem[Sei89]{sei89}
H.~Seidl.
\newblock On the finite degree of ambiguity of finite tree automata.
\newblock {\em Acta Informatica}, 26(1):527--542, 1989.

\bibitem[SS78]{salsoi78}
A.~Salomaa and M.~Soittola.
\newblock {\em Automata-Theoretic Aspects of Formal Power Series}.
\newblock Texts and Monographs in Computer Science, Springer-Verlag, 1978.

\bibitem[Wec78]{wec78}
W.~Wechler.
\newblock {\em The Concept of Fuzziness in Automata and Language Theory}.
\newblock Studien zur {A}lgebra und ihre {A}nwendungen. Akademie-Verlag Berlin,
  5. edition, 1978.

\end{thebibliography}

\end{document}